\documentclass[
a4paper,onecolumn,11pt,
]{quantumarticle}
\pdfoutput=1

\usepackage[utf8]{inputenc}
\usepackage[english]{babel}
\usepackage[T1]{fontenc}
\usepackage{hyperref}

\usepackage[pdftex]{graphicx}
\usepackage{amsmath,amsfonts,amsbsy,amssymb,amsthm,mathtools,mathrsfs}
\usepackage{epstopdf}
\usepackage{ulem,color}
\usepackage{enumitem}
\usepackage{siunitx}
\usepackage{environ}
\usepackage{xcolor} %

\usepackage{pgfplots}
\usepackage{pstool}
\usepackage{tikzsymbols}
\usetikzlibrary{plotmarks,spy}

\usepackage[font=small,labelfont=bf,justification=justified,format=plain]{caption}

\newtheorem{lemma}{Lemma}

\newtheorem{theorem}{Theorem}

\theoremstyle{definition}
\newtheorem{definition}{Definition}

\newcommand{\ketbra}[1]{|{#1}\>\mkern-4mu\<{#1}|}
\newcommand{\tr}{\textup{Tr}}
\renewcommand{\>}{\rangle}
\newcommand{\<}{\langle}
\newcommand{\N}{{\mathbb{N}}} %
\newcommand{\C}{{\mathbb{C}}} %
\newcommand{\E}{{\mathbb{E}}}

\newcommand{\llangle}{{\<\mkern-4mu\<}}
\newcommand{\rrangle}{{\>\mkern-4mu\>}}

\newcommand{\CQT}{Centre for Quantum Technologies, National University of Singapore, 3 Science Drive 2, Singapore 117543.\looseness=-1}
\newcommand{\NTU}{Nanyang Quantum Hub, School of Physical and Mathematical Sciences, Nanyang Technological University, Singapore 639673.\looseness=-1}
\newcommand{\IHPC}{A*STAR Quantum Innovation Centre (Q.InC), Institute of High Performance Computing (IHPC), Agency for Science, Technology and Research (A*STAR), 1 Fusionopolis Way, \#16-16 Connexis, Singapore, 138632, Republic of Singapore.\looseness=-1}

\newcommand{\CQuERE}{Centre for Quantum Engineering, Research and Education, TCG CREST, Sector V, Salt Lake, Kolkata 700091, India.\looseness=-1}

\usepackage{subcaption}
\begin{document}

\normalem
\newlength\figHeight 
\newlength\figWidth 

\title{Classically Spoofing System Linear Cross Entropy Score Benchmarking for Sublinear-Depth Circuit}
\author{Andrew Tanggara}
\email{andrew.tanggara@gmail.com}
\affiliation{\CQT}
\affiliation{\NTU}

\author{Mile Gu}
\email{mgu@quantumcomplexity.org}
\affiliation{\NTU}
\affiliation{\CQT}

\author{Kishor Bharti}
\email{kishor.bharti1@gmail.com}
\affiliation{\IHPC}
\affiliation{\CQuERE}

\begin{abstract}

In recent years, several experimental groups have claimed demonstrations of ``quantum supremacy'' or computational quantum advantage. 
A notable first claim by Google Quantum AI revolves around a metric called the Linear Cross Entropy Benchmarking (Linear XEB), which has been used in many quantum supremacy experiments since.
The complexity-theoretic hardness of spoofing Linear XEB, however, depends on the Cross-Entropy Quantum Threshold (XQUATH) conjecture put forth by Aaronson and Gunn, which has been disproven for sublinear depth circuits.
In the efforts on demonstrating quantum supremacy by quantum Hamiltonian simulation, a similar benchmarking metric called the System Linear Cross Entropy Score (sXES) holds firm in light of the aforementioned negative result due to its fundamental distinction with Linear XEB.
Moreover, the complexity-theoretic hardness of spoofing sXES rests on the System Linear Cross-Entropy Quantum Threshold Assumption (sXQUATH), the formal relationship of which to XQUATH is unclear.
Despite the promises offered by sXES for future demonstration of quantum supremacy, in this work we show that it can be classically simulated efficiently in certain regimes.
Particularly, we show that sXQUATH does not hold for sufficiently noisy sublinear depth circuits by constructing a classical algorithm that spoofs sXES whenever the circuit noise rate is larger than certain threshold.
\end{abstract}

\maketitle

\section{Introduction}

In $2019$, the Google quantum AI team claimed the first experimental demonstration of ``quantum supremacy,''~\cite{preskill2012quantum} or computational quantum advantage using a $53$ qubits superconducting circuit~\cite{Arute_2019}, signifying a major leap forward in practical quantum computing and challenging the extended Church-Turing thesis~\cite{arora2009computational}.
In verifying that their circuit is correctly performing a task called quantum random circuit sampling (RCS), they tested their samples using a metric called ``Linear Cross-Entropy Benchmarking (Linear XEB)''~\cite{aaronson2016complexitytheoretic,Boixo_2018,neill2018blueprint,hangleiter2019sample,eisert2020quantum,hangleiter2023computational}.
Since then, multiple RCS experiments~\cite{wu2021strong,zhu2022quantum,madsen2022quantum,morvan2023phase} have their quantum supremacy claims verified by Linear XEB method, or a variant thereof.
Skepticism on these claims have been raised, particularly by classical simulations of Google's RCS experiment~\cite{pan2022simulation,huang2020classical,liu2021closing,pan2022solving,kalachev2021classical} showing significantly shorter classical runtime compared to their initial estimation of $10,000$ years.
Moreover, theoretical results on  classical simulation of different RCS variants~\cite{gao2018efficient,barak2020spoofing,gao2021limitations,oh2023spoofing,aharonov2023polynomial,rajakumar2024polynomial} cast doubts on the complexity-theoretic hardness of spoofing Linear XEB that rests on the cross entropy quantum threshold assumption (XQUATH) conjecture proposed by Aaronson and Gunn~\cite{aaronson2020classical}.
Recent demonstration that XQUATH does not hold for sublinear depth RCS~\cite{aharonov2023polynomial}, further diminishes the legitimacy of Linear XEB as a benchmark for quantum supremacy.

A variant of Linear XEB called the System Linear Cross Entropy Score (sXES) benchmarking also aims to demonstrate quantum supremacy using a more structured family of quantum circuits called the Minimal Quantum Singular Value Transform (mQSVT) circuits, which can implement Quantum Singular Value Transform (QSVT) algorithms~\cite{gilyen2019quantum} (such as Szegedy quantum walk~\cite{szegedy2004quantum} and quantum solver for system of linear equations~\cite{Harrow_2009} and Hamiltonian simulation tasks~\cite{feynman2018simulating,lloyd1996universal}).
Structural difference between the sampling task assessed by sXES and other Linear XEB variants renders it unclear whether existing Linear XEB spoofing methods such as~\cite{gao2018efficient,gao2021limitations,aharonov2023polynomial} can be used for sXES.
In particular, it is unclear whether the ``Pauli Path'' algorithm that was shown to be able to efficiently approximate RCS output probabilities to refute XQUATH~\cite{aharonov2023polynomial} can be directly used to refute sXQUATH due to this structural difference\footnote{Essentially, the Pauli path approximation in Ref.~\cite{aharonov2023polynomial} was shown to be sufficiently close to the RCS probabilities by an analysis of the first and second moment of Haar-random expectations. 
The same analysis does not work for mQSVT circuits due to the use of multiple copies of random gates in the circuit.
Particularly, the ``orthogonality property'' central in the analysis of RCS Pauli path does not hold for mQSVT Pauli path.
Additionally, how one is required to approximate \textit{all} output $n$ bit strings except for the all zeros string $0^n$ in sXQUATH and sXES complicates the analysis even more. }.
Moreover the hardness of spoofing sXES lies upon a complexity-theoretic conjecture known as the System Linear Cross-Entropy Quantum Threshold Assumption (sXQUATH), the formal relationship  of which to XQUATH is unknown (see Appendix~\ref{App:thershold_assumptions}).
These fundamental distinctions from other Linear XEB variants thus renders sXES a promising verification method in future claims of quantum supremacy experiments.

In this work, we go beyond the technique in~\cite{aharonov2023polynomial} to show that there exists an efficient classical algorithm that approximates the experiment sufficiently well to refute sXQUATH (see Theorem~\ref{thm:sxquath_false}).
At the same time, we also show explicitly that our algorithm spoofs the sXES benchmark (see Theorem~\ref{thm:spoofing_sXES}) for noisy experiments.
Our algorithm approximates the output probability distribution of a mQSVT circuit which components are randomly sampled, which in turn gives a sufficiently high sXES.
While the mQSVT circuits offers a promising way to demonstrate quantum supremacy in the near term, our results suggest that a more robust benchmarking method is necessary.

\section{Spoofing mQSVT circuit benchmarking}

An mQSVT circuit $\mathtt{mQSVT}(U)$ (see Fig.~\ref{fig:mqsvt_circuit}) consists of $d$ ``blocks'', each containing a copy of $n+1$ qubit unitary $U$ and a copy of its conjugate $U^\dag$.
Denote the depth of $U$ as $d_U$ so that we can write $U=U_{d_U}\dots U_1$ where $U_j$ is the $j$-th layer of $U$.
These unitaries are interleaved by phase shift gates $R(\varphi)$ at the top register with carefully chosen phases (as discussed in the supplementary material of~\cite{Dong_2022}).
Samples from an mQSVT circuit are obtained from measuring the bottom $n$ registers conditioned on measurement of the top register being $0$.
The outcome probability of an $n$-bit string $x$ is therefore $p(U,x) = |\<0x|\mathtt{mQSVT}(U)|0^{n+1}\>|^2$ (for an explicit expression of the output probability of an mQSVT circuit, see Appendix~\ref{App:mqsvt_benchmark}.
Here we consider unitaries $U$ consisting only of two-qubit gates such that in each layer, every qubit register is evolved by precisely one two-qubit gate without any geometric locality assumption (hence $n+1$ is even).
This is the unitary architecture assumed for the RCS simulation result in~\cite{aharonov2023polynomial}.

\begin{figure*}
    \includegraphics[width=\linewidth]{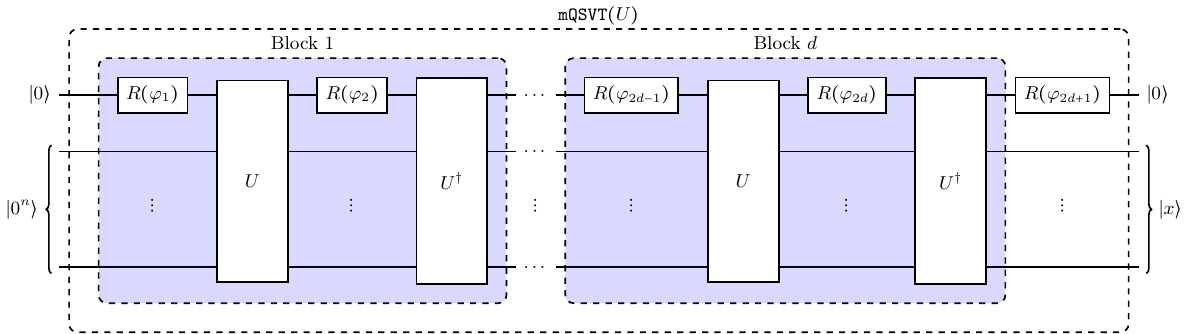}
    \caption{mQSVT circuit $\mathtt{mQSVT}(U)$ where $U$ is a random $n+1$ qubit unitary made out of 2 qubit Haar-random unitaries and $R(\varphi)$ is $Z$-rotation gate with angle $\varphi$.}
    \label{fig:mqsvt_circuit}
\end{figure*}

For a noisy mQSVT Hamiltonian simulation experiment, a benchmarking scheme similar to XEB benchmarking used in RCS experiment~\cite{Arute_2019,wu2021strong,zhu2022quantum}, called the average system linear cross-entropy score (sXES)~\cite{Dong_2022} was proposed.
For a given ideal mQSVT circuit $\mathtt{mQSVT}(U)$ and (empirically approximated) experimental probability $p_\mathrm{exp}(U,x)$ of output $x$, its sXES is given by
\begin{gather}\label{eqn:sxes_main}
    \E_U[\mathrm{sXES}(U)] = \sum_{x\neq 0^n} \E_U[p(U,x) \, p_\mathrm{exp}(U,x)] \,,
\end{gather}
where $\E_U$ is expectation over random $U$.
An experiment with high sXES indicates a high circuit fidelity as it assigns a high probability $p_\mathrm{exp}(U,x)$ to string $x$ with high ideal probability $p(U,x)$, hence higher sXES.

Computational hardness of classically spoofing sXES can be reduced to the system linear cross-entropy heavy output generation (sXHOG) sampling problem~\cite{Dong_2022}.
However, the hardness of classically solving sXHOG (with a high probability and a suitable choice of number of samples) holds only under a conjecture called the System linear cross-entropy quantum threshold assumption (sXQUATH)~\cite{Dong_2022}.
sXQUATH conjecture states that there is no polynomial-time classical algorithm taking an efficient description of $n+1$-qubit unitary $U$ and $n$ bit string $x$ as inputs and outputs an approximation $q(U,x)$ of mQSVT output probability $p(U,x)$ such that
\begin{equation}\label{eqn:sxquath_informal}
\begin{aligned}
    \textup{sXQ} &:= \mathcal{E}\Big(p(U,X),\frac{1}{2^n}\Big) -  \mathcal{E}\Big(p(U,X),q(U,X)\Big) \geq c 2^{-3n}
\end{aligned}
\end{equation}
for some constant $c$ and large enough $n$.
Here, $\mathcal{E}(f,g) := \E_{U,X}[(f(U,X)-g(U,X))^2]$ is the mean-squared error (MSE) between functions $f$ and $g$ and $\E_{U,X}$ is expectation over uniformly random variable $X\in\{0,1\}^n\backslash\{0^n\}$ and over Haar-random two-qubit gates in $n+1$-qubit unitary $U$. 
Now we present our main result below in Theorem~\ref{thm:sxquath_false}, which states that sXQUATH generally does not hold. 
Particularly for a single-block mQSVT circuit (i.e. with $d=1$) and a sublinear depth unitary $U$, one can construct a classical algorithm running in time polynomial in $n$ with a non-negligibly less MSE than the trivial approximation. 

\begin{theorem}\label{thm:sxquath_false}
    There exists an efficient classical algorithm taking an efficient description of $n+1$-qubit unitary $U$ with sublinear depth $d_U=o(n)$ and $n$ bit string $x$ as inputs and outputs an approximation $q(U,x)$ of a single-block mQSVT circuit output probability $p(U,x)$ that satisfies eqn.~\eqref{eqn:sxquath_informal}.
\end{theorem}

We further show that our algorithm spoofs mQSVT circuits with sufficiently large noise.
As mQSVT circuits with completely depolarizing noise have uniform output probability, its sXES is $\E_U[\mathrm{sXES}(U)] = 2^{-n} \sum_{x\neq0} \E_U[p(U,x)]$, indicating no correlation with the ideal probability $p(U,x)$.
Our algorithm spoofs sXES of all mQSVT circuits with depolarizing noise above a certain threshold (such that it is close to $2^{-n} \sum_{x\neq0} \E_U[p(U,x)]$).
\begin{theorem}\label{thm:spoofing_sXES}
    There exists an efficient classical algorithm spoofing sXES for all noisy single-block mQSVT circuit with $n+1$-qubit unitary $U$ such that its sXES is at most $2^{-n}(\sum_{x\neq0} \E_U[p(U,x)] + c^{d_U})$ for some constant $c>0$.
\end{theorem}

If we consider mQSVT circuit corrupted with depolarizing noise with noise strength $\gamma\in[0,1]$ on each of its register in each layer, then for sufficiently large $\gamma$ its sXES score is going to be less than $2^{-n}(\sum_{x\neq0} \E_U[p(U,x)] + c^{d_U})$.
Theorem~\ref{thm:spoofing_sXES} indicates that the sXES all such noisy mQSVT circuit is spoofable by our algorithm.

Classical algorithm that refutes sXQUATH and spoofs sXES benchmark in Theorem~\ref{thm:sxquath_false} and Theorem~\ref{thm:spoofing_sXES} above is from a family of algorithms called the Pauli path algorithms.
The proof (sketch) of Theorem~\ref{thm:sxquath_false} and Theorem~\ref{thm:spoofing_sXES} will be given in Section~\ref{sec:refuting_sxquath} and Section~\ref{sec:spoofing_sxes} after we lay out the Pauli path framework below.
We will also discuss how the existing instances of Pauli path algorithm used to classically simulate quantum circuits~\cite{bremner2017achieving,gao2018efficient,aharonov2023polynomial} are not directly applicable to mQSVT circuit, mainly due to the existence of multiple copies of random unitaries.

\section{Pauli path algorithm}

Given an $n$-qubit unitary quantum circuit $C=C_d C_{d-1} \dots C_1$ (where $C_j$ is the $j$-th layer of circuit $C$) with a fixed input $|0^n\>$ and a computational basis measurement, Pauli path algorithm classically computes the output probabilities by expanding density matrices at every layer in terms of normalized $n$-qubit Pauli matrices $\mathcal{P}_n = \{I/\sqrt{2},X/\sqrt{2},Y/\sqrt{2},Z/\sqrt{2}\}^{\otimes n}$.
At the input, we get $\ketbra{0^n} = \sum_{s_1\in\mathcal{P}_n} s_1 \tr(s_1\ketbra{0^n})$.
We can substitute this expansion to the density matrix after the first layer, $\rho_1:=C_1\ketbra{0^n}C_1^\dag$ then expand it in the same manner to get $\rho_1=\sum_{s_1,s_2\in\mathcal{P}_n} s_2\tr(s_2 C_1 s_1 C_1^\dag)\tr(s_1\ketbra{0^n})$.
Repeating this for the remaining layers and for measurement $\ketbra{x}$, we obtain \textit{transition amplitudes} in the coefficients of the expansion that correspond to each step in the circuit.
Transition amplitude at the input and the measurement steps are defined by $\llangle s_1|0^n\rrangle := \tr(s_1\ketbra{0^n})$ and $\llangle x|s_{d+1}\rrangle := \tr(\ketbra{x}s_{d+1})$, respectively.
Whereas the transition amplitude corresponding to the $j$-th layer is $\llangle s_{j+1}|C_j|s_j \rrangle := \tr(s_{j+1}C_js_jC_j^\dag)$, where we call $s_j$ an \textit{input Pauli} and $s_{j+1}$ an \textit{output Pauli}.
A sequence of normalized $n$-qubit Paulis $\mathbf{s}=s_1,s_2,\dots,s_{d+1}$ is called a \textit{Pauli path}.
Then, the probability of obtaining a measurement outcome $n$ bit string $x$ is
\begin{gather}\label{eqn:rcs_output_proba_pauli_path}
    |\<x|C|0^n\>|^2 = \sum_{s_1,s_2,\dots,s_{d+1}\in\mathcal{P}_n} f(C,\mathbf{s},x) \;,
\end{gather}
where each Pauli path defines a \textit{Fourier coefficient} 
\begin{equation}
    f(C,\mathbf{s},x) = \llangle x|s_{d+1}\rrangle  \llangle s_{d+1}|C_d|s_d\rrangle \dots \llangle s_2|C_1|s_1\rrangle  \llangle s_1|0^n \rrangle \;.
\end{equation}

An important property of the Pauli path algorithm in approximating RCS (with the circuit architecture as mQSVT circuit unitary explained earlier) is the orthogonality property~\cite{harrow2009random,aharonov2023polynomial} which states that
\begin{equation}\label{eqn:orthogonality_fourier_rcs}
\begin{gathered}
    \E_C\Big[f(C,\mathbf{s},x) f(C,\mathbf{r},x)\Big] = 0 \;,
\end{gathered}
\end{equation}
for Pauli paths $\mathbf{s}\neq\mathbf{r}$ and $\E_C$ expectation over random circuit $C$.
This is due to independently sampled two-qubit gates in $C$ that decomposes the expectation over circuit $C$ as a product of Haar-random expectations over these two-qubit gates.
As each two-qubit gate $V$ appears twice, once in $f(C,\mathbf{s},x)$ and once in $f(C,\mathbf{r},x)$, we get expressions of the form $\tr( \E_V[V^{\otimes2}s\otimes r (V^\dag)^{\otimes2}] s'\otimes r' )$ which evaluates to 0 when $s\neq r$ or $s'\neq r'$.

Unfortunately, orthogonality condition~\eqref{eqn:orthogonality_fourier_rcs} does not hold in general for mQSVT circuits due to a random unitary $U$ appearing $2d$ number of times in a $d$-block mQSVT circuit (see 
Appendix~\ref{App:single_block_mQSVT_pauli_path} for the case of $d=1$).
This requires the more complicated analysis of higher moment two-qubit Haar-random expectation, which unlike the second moment, is negative for some choices of Paulis.
Now we define an mQSVT Pauli path $\mathbf{s} = \mathbf{s}^{(1)},\mathbf{s}^{(2)},\dots, \mathbf{s}^{(d)}$ is a sequence of sub-paths $\mathbf{s}^{(k)} = s_1^{(k)},\dots,s_{d_U+1}^{(k)},\Tilde{s}_{d_U+1}^{(k)},\dots,\Tilde{s}_1^{(k)}$ that goes through unitary $U$ and $U^\dag$ in the $k$-th block.
So, $s_j,s_{j+1}$ are the input and output Paulis for layer $U_j$, respectively. 
Whereas $\Tilde{s}_{j+1},\Tilde{s}_j$ are input and output Paulis for $U_j^\dag$.
Since phases $\varphi_1,\dots,\varphi_{2d+1}$ in an mQSVT circuit are fixed, we can instead simply absorb rotation gates $R(\varphi_1)$ and $R(\varphi_{2d+1})$ at the beginning and end of the circuit into the preparation and measurement, respectively.
So for a Pauli path $\mathbf{s}$ through a mQSVT circuit with unitary $U$, we can define an mQSVT circuit Fourier coefficient as 
\begin{equation}
    F(U,\mathbf{s},x) = \llangle0x|\Tilde{s}_1^{(d)}\rrangle \prod_{k=1}^d F_k(U,\mathbf{s}^{(k)}) \llangle s_1^{(1)}|0^{n+1}\rrangle \;,
\end{equation}
where $F_k(U,\mathbf{s}^{(k)})$ is the $k$-th \textit{block Fourier coefficient} that contains the product of transition amplitudes of each layer of $U$ and $U^\dag$ and rotation gates $R(\varphi_{2k})$ and $R(\varphi_{2k-1})$ (for an explicit expression, see 
Appendix~\ref{App:mQSVT_pauli_path}).
Hence the probability of output $x$ from an mQSVT circuit with unitary $U$ is 
\begin{equation}\label{eqn:fourier_expansion_mqsvt}
    p(U,x) = \sum_{\mathbf{s}} F(U,\mathbf{s},x) \;.
\end{equation}

\subsection{Pauli Path Simulation for mQSVT Circuit}

As we have mentioned, both in showing that sXQUATH does not hold (Theorem~\ref{thm:sxquath_false}) and in spoofing sXES (Theorem~\ref{thm:spoofing_sXES}) for single-block mQSVT circuit, we use the Pauli path algorithm to approximate the output probabilities of an mQSVT circuit.
Particularly for a given unitary $U$ and outcome $x$, we approximate the mQSVT output probability by
\begin{equation}\label{eqn:pauli_path_for_spoofing}
    q(U,x) = \frac{1}{2^n} + F(U,\mathbf{r},x)
\end{equation}
for a Pauli path $\mathbf{r} = (Z\otimes I^{\otimes l-1} \otimes Z \otimes I^{\otimes n-l} , Z\otimes I^{\otimes n} , \dots , Z\otimes I^{\otimes n})$, which can be computed in time polynomial in $n$.
From the second layer onwards, all Paulis in path $\mathbf{r}$ are identity except for the Pauli $Z$ at the first register.
However, $r_1=Z\otimes I^{\otimes l-1} \otimes Z \otimes I^{\otimes n-l}$ depends on the architecture in the first layer of $U$.
Namely we set $l\in\{2,\dots,n+1\}$ such that the first qubit is coupled with the $l+1$-th qubit in the first layer $U_1$ by two-qubit gate $V_1$.
Thus the input Paulis for this two-qubit gate $V_1$ is $Z\otimes Z$.
For example, if in the two-qubit gate $V_1$ that operates on the first qubit couples it with the 5th qubit then $r_1= Z\otimes I^{\otimes3} \otimes Z \otimes I^{\otimes n-4}$.
We also denote the two-qubit gate in the $j$-th layer of $U$ that couples the first register with some other register as $V_j$.
Proof outline of how this Pauli path approximation refutes sXQUATH and spoofs sXES for single-block mQSVT circuit are given in the methods section, while a full technical proof can be found in Appendix~\ref{App:sxquath_false} and Appendix~\ref{App:spoofing_sXES}.
Here we discuss the intuition on why this Pauli path simulation is a sufficiently good approximation to obtain the aforementioned results.

By using approximation $q(U,x)$ defined in eqn.~\eqref{eqn:pauli_path_for_spoofing}, we found that the expected approximation error $\mathcal{E}(p(U,X),q(U,X)) = \E_{U,X}[(p(U,X)-q(U,X))^2]$ in the sXQUATH expression (eqn.~\eqref{eqn:sxquath_informal}) is sufficiently smaller than the approximation error $\mathcal{E}(p(U,X),\frac{1}{2^n})$ by the trivial algorithm outputting uniformly random $x$.
This is mainly due to approximation $q(U,X)$ having a sufficiently high correlation with $p(U,X)$ over the outcomes $X$ and random unitary $U$, which is quantified by $\E_{U,X}[p(U,X)q(U,X)]$.
Note that this correlation is proportional to the sXES benchmarking in eqn.~\eqref{eqn:sxes_main}, since the expectation is uniform over all nonzero output $x\neq 0^n$.
Simply put, the reason why approximation $q(U,x)$ works is because for each $U$ it first assigns a uniform probability $2^{-n}$ to all outputs, then adjusts the output probabilities to be slightly closer to the true probability distribution $p(U,x)$ by computing the Fourier coefficient $F(U,\mathbf{r},x)$.
Fourier coefficient $F(U,\mathbf{r},x)$ essentially captures some information of how more or less likely outcome $x$ is compared to the uniform distribution $2^{-n}$ by evaluating the transition amplitudes of two-qubit unitaries operating on the first register of the mQSVT circuit $\mathtt{mQSVT}(U)$.
It can be shown (see methods) that the Fourier coefficient $F(U,\mathbf{r},x)$ used in the approximation $q(U,x)$ has a sufficiently large ``overlap'' with the true probability $p(U,x)$ by expanding $p(U,x)$ in terms of Fourier coefficient of all possible Pauli paths as in eqn.~\eqref{eqn:fourier_expansion_mqsvt}.

\subsection{Pauli Path Algorithm and Haar-random Unitary Moment Matrix}

Here we discuss Haar-Random unitary moment matrix, which plays a central role in the analysis of how good a Pauli path algorithm approximates a random quantum circuit output distribution.
For a given $2$-qubit Pauli matrices $p_1,\dots,p_t \in \mathcal{P}_2$, its expectation over product of $t$ Haar-random $2$-qubit unitary $V$ is given by
\begin{gather}
    \E_V\Big[ V^{\otimes t} p_1 \otimes\dots\otimes p_t (V^\dag)^{\otimes t} \Big] \;.
\end{gather}
This quantity is extensively studied in~\cite{harrow2009random} and used in~\cite{Aharonov_Gao_Landau_Liu_Vazirani_2022} to show the properties of the transition amplitudes in the Fourier coefficients of Pauli paths.
For our purposes we focus on the case of $2$-qubit Paulis $p_1,\dots,p_t \in \mathcal{P}_2$ and $2$-qubit Haar-random unitary $V$.

Now consider a matrix $G$ with each of its entry defined by $t$ two-qubit Paulis $\mathbf{p} = p_1 ,\dots, p_t$ and $\mathbf{q}= q_t ,\dots, q_t$
\begin{align}\label{eqn:haar_unitary_moment_matrix_main}
    G(\mathbf{p};\mathbf{q}) = \tr\bigg( \E_V\Big[ V^{\otimes t} p_1 \otimes\dots\otimes p_t (V^\dag)^{\otimes t} \Big] q_t \otimes\dots\otimes q_t \bigg) \;.
\end{align}
Matrix $G$ is what is known as the two-qubit \textit{Haar-random unitary moment matrix} of order $t$.
For a random $N$-qubit circuit $U$ constructed from Haar-random two-qubit unitaries, then the expectation of Pauli path transition amplitudes can be expressed as a product of expectation of transition amplitudes $\E_V[\llangle p |V|p' \rrangle] = \E_V[\tr(V p V^\dag p')] = G(p;p')$ for two-qubit Haar random unitary $V$ and normalized two-qubit Paulis $p,p'$.
This is precisely the key analysis being done in analyzing the Pauli path algorithm for RCS~\cite{aharonov2023polynomial}.

However in an mQSVT circuit with $d$ blocks we have $2d$ copies of the same Haar-random unitary $U$ (including its conjugate $U^\dag$), which means that each two-qubit Haar-random gate is being used $2d$-times throughout the circuit.
An mQSVT circuit with $d=1$ block contains a single $U$ and its conjugate $U^\dag$.
Thus for Pauli path $\mathbf{s}$ and output $x$, Haar-random expectation of its Fourier coefficient $\E_U[F(U,\mathbf{s},x)]$ can be decomposed as a product of expectations of the transition amplitudes $\E_V[\llangle p|V|q\rrangle \llangle q'|V^\dag|p'\rrangle]$ where $V$ is a random two-qubit gate in $U$ and $U^\dag$.
Since $\llangle q'|V^\dag|p'\rrangle = \llangle p'|V|q'\rrangle$ therefore one can write this two-qubit Haar-random expectation as an entry of the Haar-random moment matrix of order $t=2$
\begin{equation}
\begin{aligned}
    G(p,p';q,q') &= \tr\Big( \E_V\big[V^{\otimes2} p\otimes p' (V^\dag)^{\otimes2}\big] q\otimes q' \Big) \\
    &= \E_V[\llangle p|V|q\rrangle \llangle q'|V^\dag|p'\rrangle] \;.
\end{aligned}
\end{equation}
for each two-qubit Haar random unitary $V$ in $U$.

\section{Refuting sXQUATH}\label{sec:refuting_sxquath}

Here we give an outline of the proof  of Theorem~\ref{thm:sxquath_false}. 
The detailed, more technical proof is deferred to 
Appendix~\ref{App:sxquath_false}.
For readability, we omit the normalization factors of the Paulis in the paths, e.g. we write $Z$ when we mean $Z/\sqrt{2}$.

Consider an mQSVT Pauli path approximation for probability of outcome $x$,
\begin{gather}\label{eqn:pauli_path_classical_estimation_main}
    q(U,x) = \frac{1}{2^n} + F(U,\mathbf{r},x) \;,
\end{gather}
for Pauli path $\mathbf{r} = (Z\otimes I^{\otimes l-1} \otimes Z \otimes I^{\otimes n-l} , Z\otimes I^{\otimes n} , \dots , Z\otimes I^{\otimes n})$.
Since there is only a single block, we write $\mathbf{r} = r_1,\dots,r_{d_U+1},\Tilde{r}_{d_U+1},\dots,\Tilde{r}_1$ where the $r_j$'s and $\Tilde{r}_j$ are input/output Paulis in $U$ and $U^\dag$, respectively.
From the second layer onwards, all Paulis in path $\mathbf{r}$ are identity except for the Pauli $Z$ at the first register.
However, $r_1=Z\otimes I^{\otimes l-1} \otimes Z \otimes I^{\otimes n-l}$ depends on the architecture in the first layer of $U$.
Namely we set $l\in\{2,\dots,n+1\}$ such that the first qubit is coupled with the $l+1$-th qubit in the first layer $U_1$ by two-qubit unitary $V$.
Thus the input Paulis for this two-qubit gate is $Z\otimes Z$.
For example, if in the first layer the two qubit gate that operates on the first qubit couples it with the 5th qubit then $r_1= Z\otimes I^{\otimes3} \otimes Z \otimes I^{\otimes n-4}$.
We also denote the two-qubit gates in the $j$-th layer of $U$ as $V_{j,i}$, for some arbitrary indexing $i\in[(n+1)/2]$ (as the coupling by these gates are arbitrary, i.e. gates $V_{1,1}$ and $V_{2,1}$ may operate on different pair of registers).

Now by using $q(U,x)$ in sXQUATH eqn.~\eqref{eqn:sxquath_informal} as an approximation to $p(U,x)$ and by expanding the squares and simplifying the terms we get
\begin{equation}\label{eqn:sxquath_simplified_main}
\begin{aligned}
    &\textup{sXQ} \\
    &= \frac{1}{2^n-1} \sum_{x\neq 0^n} \bigg( -\underbrace{\frac{2}{2^n}\E_U\Big[F(U,\mathbf{r},x)\Big]}_{(i)} - \underbrace{\E_U\Big[F(U,\mathbf{r},x)^2\Big]}_{(ii)} + \underbrace{2\sum_{\mathbf{s}\neq\mathbf{r}} \E_U\Big[ F(U,\mathbf{s},x)\, F(U,\mathbf{r},x)}_{(iii)} \Big] \bigg) \,,
\end{aligned}
\end{equation}
where we also expand $p(U,x)$ in terms of the mQSVT Fourier coefficients.
For notational simplicity, from now on we will omit the tensor "$\otimes$" in denoting two-qubit Paulis that is an input or output to the same two-qubit gate.
For example, by $V^{\otimes2}(ZZ\otimes II)(V^\dag)^{\otimes2} IZ\otimes II$ we mean $V^{\otimes2} (Z\otimes Z\otimes I\otimes I)(V^\dag)^{\otimes2} I\otimes Z\otimes I\otimes I$.

We  treat the terms $(i),(ii),(iii)$ separately in eqn~\eqref{eqn:sxquath_simplified_main}.
However the main idea is similar.
Since the two-qubit gates $\{V_{j,i}\}_{j,i}$ in $U$ are sampled independently, the expectation $\E_U$ can be decomposed as a product of expectations over each two-qubit gate $\E_{V_{j,i}}$ and the transition amplitudes for the input, output, and the rotation gate (see 
Appendix~\ref{App:single_block_mQSVT_pauli_path}).
More explicitly, this decomposition allows us to express expectation $\E_U$ as the 
product of order-$t$ expectations of two-qubit Haar unitary 
\begin{equation}\label{eqn:two_qubit_haar_expectation}
    \tr\,\E_{V_{j,i}}\Big[ V_{j,i}^{\otimes t} p_{j,i}^{(1)}\otimes\dots\otimes p_{j,i}^{(t)} (V_{j,i}^\dag)^{\otimes t} p_{j+1,i}^{(1)}\otimes\dots\otimes p_{j+1,i}^{(t)} \Big] \,,
\end{equation}
for some positive integer $t$ and two-qubit paulis $p_{j,i}^{(l)},p_{j+1,i}^{(l)}$, where $t=2$ for term $(i)$ and $t=4$ for terms $(ii)$ and $(iii)$.
Note that Paulis in  path $\mathbf{r}$ at the second to $n+1$-th qubit registers are all identity, except for the Paulis at the input $r_1$.
Namely, $r_2,\dots,\Tilde{r}_1$ equal to $Z\otimes I^{\otimes n}$.
As a consequence, most of these expectations are of the form $\tr\E[(VIIV^\dag II)^{\otimes t}]$, which takes the value of $1$ (up to some normalization).
Of course, the only exceptions are the expectations with respect to the two-qubit gates that are interacting with the first qubit register, since the first register Paulis are all $Z$.
When there are non-identity Paulis, expectation~\eqref{eqn:two_qubit_haar_expectation} can take any real values depending on the choice of Paulis and order $t$.
This is precisely where the non-trivialities occur, thus requiring a more extensive analysis on the terms $(i),(ii),(iii)$ above.
Since the mQSVT circuit that we consider only have one block, we omit the superscript for the Paulis in $\mathbf{r}$ and simply write $r_j$ and $\Tilde{r}_j$.

Since in term $(i)$ each two-qubit unitary $V_{j,i}$ appears only two times (in $U$ and $U^\dag$), we only have second moment expectations, allowing us to use the results in~\cite{harrow2009random,aharonov2023polynomial} to determine its values (see 
Appendix~\ref{App:single_block_mQSVT_pauli_path}).
Here we only need to consider transition amplitude for the two-qubit gate $V$ in the first layer operating on the first register and the $l$-th register, the only gate in the first layer which input pauli is $ZZ$.
Since its transition amplitude $\tr\,\E_V[V^{\otimes2} ZZ\otimes ZZ (V^\dag)^{\otimes2} IZ\otimes IZ]=0$, term $(i)$ is 0.

In terms $(ii)$ and $(iii)$, each Haar random two-qubit unitary $V_{j,i}$ appears four times, twice in each Fourier coefficient (as each Fourier coefficient contains $U$ and $U^\dag$).
Hence we consider order $t=4$ of expectation~\eqref{eqn:two_qubit_haar_expectation}, $\E_V[V^{\otimes2} r\otimes \Tilde{r} \otimes s\otimes \Tilde{s} (V^\dag)^{\otimes2} r'\otimes \Tilde{r}' \otimes s'\otimes \Tilde{s}']$ hence the results~\cite{harrow2009random,aharonov2023polynomial} are no longer useful.
We instead make a detour to unitary random matrix theory of Weingarten calculus which studies the expansion of Haar-random unitary moment operation~\cite{gu2013moments,collins2022weingarten} (for applications in quantum theory, see also~\cite{roberts2017chaos,liu2021moments}).
Using the Weingarten calculus framework, we can express the Haar unitary expectations~\eqref{eqn:two_qubit_haar_expectation} in terms of Weingarten functions and permutation operators.
By evaluating the values of the Weingarten functions and analyzing permutations over a tensor product of four two-qubit Paulis (see 
Appendix~\ref{App:single_block_mQSVT_pauli_path}),
we can obtain the relevant values of expectations~\eqref{eqn:two_qubit_haar_expectation}.
As now the expectation~\eqref{eqn:two_qubit_haar_expectation} can take negative values, one needs to take great care that term $(ii)$ and $(iii)$ are sufficiently large in order to obtain $\mathrm{sXQ}=\Omega(2^{-3n})$.
We will show below that our chosen Pauli path $\mathbf{r}$ does this.

For term $(ii)$, the analysis is relatively simple since the expectations~\eqref{eqn:two_qubit_haar_expectation} for gates with non-identity input/output Paulis turns out to be all positive (see 
Appendix~\ref{App:sxquath_false}).
Note that there are precisely $d_U$ many such gates, i.e. those that are operating on the first qubit.
By also accounting for the normalization of $1/\sqrt{2}$ for each Pauli, we get $\E_U[F(U,\mathbf{r},x)^2] \geq 2^{-2n}\alpha^{d_U}$ for some constant $\alpha\in(0,1)$ independent of $n$ and $d_U$.

The analysis of term $(iii)$ is the most involved part of the proof since it involves the Fourier coefficient $F(U,\mathbf{s},x)$ for all Pauli path $\mathbf{s}$.
The trick here is to see for which Pauli path $\mathbf{s}$ the expectation $\E_U[F(U,\mathbf{s},x)F(U,\mathbf{r},x)]$ is negative, positive, and equals to 0.
This can be determined by the output transition amplitude $\llangle 0x|\Tilde{s}\rrangle$ in $F(U,\mathbf{s},x)$ and by the two-qubit Haar unitary expectations~\eqref{eqn:two_qubit_haar_expectation}.
The analysis of the latter is made slightly simpler by our choice of Pauli path $\mathbf{r}$ since we only need to analyze three types of expectations~\eqref{eqn:two_qubit_haar_expectation}, namely those where the input (output) two-qubit Paulis in $F(U,\mathbf{r},x)$ are either $ZZ$ ($ZI$) or $ZI$ ($ZI$) or $II$ ($II$).
Then we evaluate the expectations using the Weingarten function expansion (see 
Appendix~\ref{App:single_block_mQSVT_pauli_path}) 
to identify the non-zero values.
It turns out that $\E_U[F(U,\mathbf{s},x)F(U,\mathbf{r},x)]$ is non-zero only if the input/output two-qubit Paulis of $\mathbf{s}$ to each two qubit gate $V_{j,i}$ must be from the set $\{II,ZZ,ZI,IZ,PX,PY\}$ for $P\in\{I,X,Y,Z\}$.
Furthermore since $\<b|P|b\>=0$ for $P\in\{X,Y\}$ and $b\in\{0,1\}$, the paths $\mathbf{s}$ with Pauli $X$ or $Y$ at the output transition amplitude has $F(U,\mathbf{s},x)=0$, so we do not need to consider such Pauli paths.
The other complicated analysis comes from the fact that some output transition amplitude $\llangle 0x|\Tilde{s}\rrangle$ can be negative because $\<1|Z|1\>=-\frac{1}{2}$.
However we can pull the sum $\sum_{x\neq0^n}$ inside the brackets in eqn.~\eqref{eqn:sxquath_simplified_main} and instead consider $\sum_{x\neq0^n}\sum_{\mathbf{s}\neq\mathbf{r}} \E_U[F(U,\mathbf{s},x)F(U,\mathbf{r},x)]$.
From here we can separate out the negative terms and positive terms according to $x\neq0^n$ and the measurement layer of Paulis in $\mathbf{s}$.
The reader can find the detailed calculation in 
Appendix~\ref{App:sxquath_false}, 
but essentially the negative terms from to transition amplitude at the measurement layer $\llangle 0x|\Tilde{s}\rrangle$ and from the expectations~\eqref{eqn:two_qubit_haar_expectation} either cancel each other out or small enough with respect to the positive terms in the sum over $x\neq0$ and $\mathbf{s}\neq\mathbf{r}$.
This rather involved calculation gives us $\sum_{x\neq0^n}\sum_{\mathbf{s}\neq\mathbf{r}} \E_U[F(U,\mathbf{s},x)F(U,\mathbf{r},x)] \geq 2^{-n}\beta^{d_U}$ for some constant $\beta\in(0,1)$ independent of $n$ and $d_U$.

Finally, we can substitute back in the values we obtained for terms $(i),(ii),(iii)$ to eqn.~\eqref{eqn:sxquath_simplified_main} to get 
\begin{equation}
    \mathrm{sXQ} \geq \frac{1}{2^n-1}\Big( (2^n-1)2^{-2n}\alpha^{d_U} + 2^{-n}\beta^{d_U} \Big) \;.
\end{equation}
Thus for $U$ with depth sublinear in $n$, i.e. $d_U=o(n)$, we get $\mathrm{sXQ} = \Omega(2^{-2n}(\alpha+\beta)^n) = \Omega(2^{-3n})$ which refutes sXQUATH for single-block mQSVT circuit, proving Theorem~\ref{thm:sxquath_false}.

\section{Spoofing sXES benchmark}\label{sec:spoofing_sxes}

The same Pauli path that we use to refute sXQUATH above can also be used to spoof the average sXES benchmark~\eqref{eqn:sxes_main} as stated in Theorem~\ref{thm:spoofing_sXES}.
We use the same Pauli path approximation $q(U,x)$ (eqn.~\eqref{eqn:pauli_path_classical_estimation_main}) in place of experimental probability $p_\mathrm{exp}$.
Details are given in 
Appendix~\ref{App:spoofing_sXES}, 
but the proof idea is similar to the proof of Theorem~\ref{thm:sxquath_false} above.
First we expand $p(U,x)$ in terms of fourier coefficients $F(U,\mathbf{s},x)$ to obtain
\begin{equation}
\begin{aligned}
    \mathrm{sXES} &= \sum_{x\neq0^n} \E_U\Big[ \frac{p(U,x)}{2^n}\Big] + \sum_{\mathbf{s}\neq\mathbf{r}} \E_U\Big[F(U,\mathbf{s},x) \, F(U,\mathbf{r},x) \Big] \;.
\end{aligned}
\end{equation}
Note that the second term above is equal to term $(iii)$ in eqn.~\eqref{eqn:sxquath_simplified_main}, which can be lower bounded as $\sum_{x\neq0^n}\sum_{\mathbf{s}\neq\mathbf{r}} \E_U[ F(U,\mathbf{s},x)\, F(U,\mathbf{r},x) ] \leq 2^{-n}c^{d_U}$ for some constant $c\in(0,1)$ independent of $n$ and $d_U$.
By also noting that $\E_U[p(U,x)] = \frac{1}{2^{n+1}}$ for all $x$, we can lower bound sXES score as
\begin{equation}
\begin{aligned}
    \mathrm{sXES} &\geq \frac{2^n-1}{2^{2n+1}} +
    2^{-n}c^{d_U} \approx \frac{1+c^{d_U}}{2^n}
\end{aligned}
\end{equation}
for some constant $c\in(0,1)$ independent of $n$ and $d_U$.

\section{Discussion}

Linear XEB has been the method of choice in verifying sampling-based quantum supremacy claim in many experiments.
However, results on efficient classical sampling algorithms spoofing Linear XEB and a result that disproves the complexity-theoretic assumption which hardness of Linear XEB rests upon cast doubts on its robustness.
In an effort for a more robust quantum supremacy benchmarking scheme, a variant of Linear XEB called sXES has been proposed for quantum Hamiltonian simulation experiments.
Its promise relies upon the quantum circuit that it samples from being structurally distinct from circuits used for other Linear XEB, hence preventing current spoofing results to be applied directly, while also relying on a different complexity-theoretic foundation.
We have shown that sXES benchmarking for experiments corruped by high amount of noise is susceptible to spoofing.
At the same time, our result also shows the frailty of the complexity-theoretic foundation that sXES rests upon, particularly for sublinear-depth circuits.
Our negative results extend the seminal result~\cite{aharonov2023polynomial} that spoofs Linear XEB for unitary random circuit sampling and further support the need for a novel benchmarking task with a stronger complexity-theoretic guarantee for any future quantum supremacy experiments.
In particular, such guarantee needs to rule out any possibility of spoofing by the state-of-the-art algorithms used in our result as well as those used in previous spoofing results.

\section*{Acknowledgements}
The authors thank Dax Koh and Rahul Jain for interesting discussions.
AT is supported by CQT PhD Scholarship. This work is supported by the National Research Foundation of Singapore through the NRF Investigatorship Program (Award No. NRF-NRFI09-0010), the Singapore Ministry of Education Tier 1 Grant RG91/25 and RT4/23, the National Quantum Office, hosted in A*STAR, under its Centre for Quantum Technologies Funding Initiative (S24Q2d0009), and A*STAR C230917003. 
\bibliographystyle{quantum}
\bibliography{references}
\appendix

\section{Random Circuit Sampling Benchmark}\label{App:RCS_benchmarking}

Linear XEB benchmarking for an RCS experiment quantifies how close the noisy circuit implementation is to the ideal circuit.
It is done by analyzing samples taken from the noisy circuit to obtain an empirical estimation of the circuit's output probability distribution.

\begin{definition}\label{def:XEB}
    Linear cross-entropy benchmark (XEB) of a noisy implementation of an $n$ qubit unitary $U$ is defined as
    \begin{gather}\label{eqn:xeb}
        \mathrm{XEB}(U) = 2^n \sum_{x\neq 0^n} p(U,x) \, p_\mathrm{exp}(U,x) - 1
    \end{gather}
    where $p_\mathrm{exp}(U,x)$ is the output probability from the noisy implementation of circuit $U$.
\end{definition}

To argue for the quantum supremacy from an RCS experiment, a noisy circuit used for the experiment with sufficiently high Linear XEB (on average over some distribution that circuit $U$ is sampled from) implies its ability to solve a problem that is supposedly hard for classical computers.
Aaronson and Gunn proposed the XHOG problem~\cite{aaronson2020classical}, which demands an algorithm to output a set of strings that have sufficiently high probabilities from a given circuit $U$.

\begin{definition}\label{def:XHOG}
    Linear cross-entropy heavy output generation (XHOG) problem:
    For a given $n$ qubit unitary $U$ and some $b>1$, output $k$ many distinct non-zero $n$ bit strings $\{x_1,\dots,x_k\}\subseteq\{0,1\}^n\backslash\{0^n\}$ such that they satisfy
    \begin{gather}
        \frac{1}{k} \sum_{j=1}^k p(U,x_j) > \frac{b}{2^n} \;.
    \end{gather}
\end{definition}

However, showing the hardness of solving XHOG classically remains a problem.
Aaronson and Gunn have shown that if there is no classical algorithm that can approximate the output probability of the all zero string $0^n$ from a random circuit $U$ with an error that is slightly less than that of a trivial algorithm, then XHOG is a hard problem for classical computers to solve with some high probability for a suitable choice of $k$ and $b$. 

\begin{definition}\label{def:xquath}
    Linear cross-entropy quantum threshold assumption (XQUATH):
    There is no polynomial-time classical algorithm $\mathtt{C}$ that given a description of an $n$ qubit quantum circuit $U$ with depth $d$ and input bits $0^n$ and output bits $0^n$ computes the output probability $p_\mathtt{C}(U,0^n)$ such that the following holds for all $n\in\N$
    \begin{gather}
        \textup{XQ} = 2^{2n} \bigg( \E_U\Big[\Big(p(U,0^n)-\frac{1}{2^n}\Big)^2\Big] - \E_U\Big[\Big(p(U,0^n)-q_{\mathtt{C}}(U,0^n)\Big)^2\Big] \bigg) = \Omega\Big(\frac{1}{2^n}\Big)
    \end{gather}
    where and $p(U,0^n)$ is the probability of output $0^n$ on mQSVT circuit for a given unitary $U$.
    The expectations are over Haar-random two qubit unitary gates that makes up $U$ where each qubit goes through a two-qubit gate at each layer.
\end{definition}

The Pauli path algorithm is used in~\cite{aharonov2023polynomial} as an efficient classical approximation to the output probabilities of a random unitary circuit $U$ with sufficiently low MSE to show that XQUATH does not hold for circuit $U$ of sublinear depth.

\section{Hardness of classically spoofing mQSVT circuit benchmarking}\label{App:mqsvt_benchmark}

\begin{figure*}
    \includegraphics[width=\linewidth]{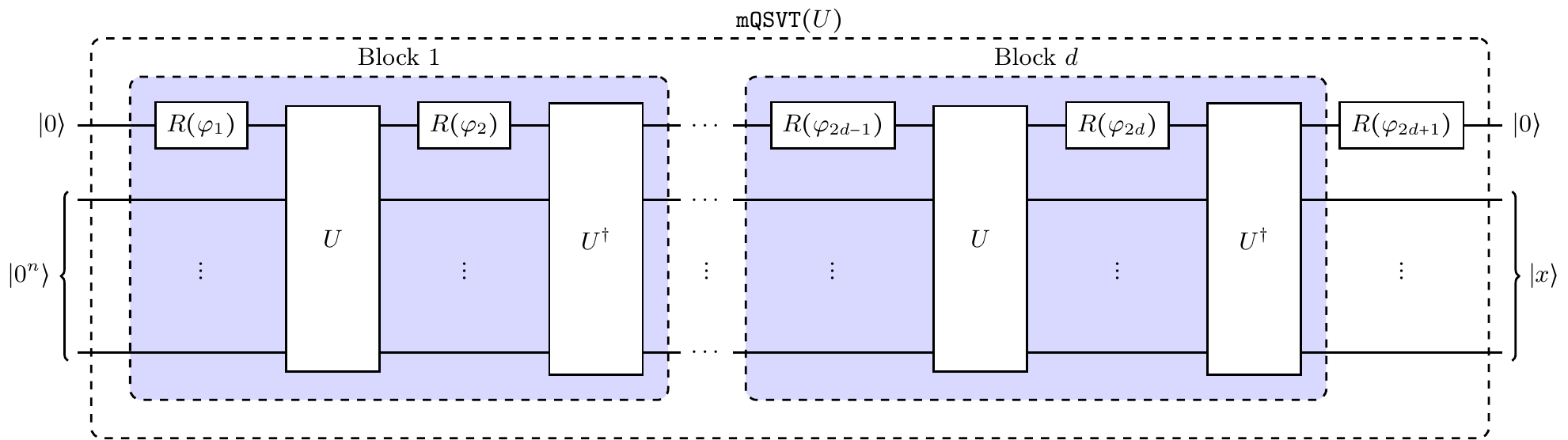}
    \caption{mQSVT circuit $\mathtt{mQSVT}(U)$ where $U$ is a random $n+1$ qubit unitary made out of 2 qubit Haar-random unitaries and $R(\varphi)$ is $Z$-rotation gate with angle $\varphi$.}
    \label{fig:mqsvt_circuit2}
\end{figure*}

Formally, an mQSVT circuit consists of $d$ ``blocks'', each containing a copy of $n+1$ qubit unitary $U$ and a copy of its conjugate $U^\dag$ interleaved by phase shift gates $R_z(\varphi)$ at the first register (see Fig.~\ref{fig:mqsvt_circuit2}).
The probability of output $x\in\{0,1\}^n$ from an mQSVT circuit with $d$ layers and unitary $U$ is
\begin{gather}
    p(U,x) = \Big|\<0x| \Big(R_z(\varphi_{2d+1})\otimes I^n\Big) \prod_{k=0}^d \bigg( U^\dag \Big(R_z(\varphi_{2k+2})\otimes I^n\Big) U \Big(R_z(\varphi_{2k+1})\otimes I^n\Big) \bigg) |0^{n+1}\>\Big|^2 \;,
\end{gather}
where $R_z(\varphi) = e^{i\varphi Z}$.
So if $d_U$ is depth of circuit $U$, then the total depth of the mQSVT circuit is $d_\mathrm{mQSVT} = d(2d_U+2)+1$.

In an Hamiltonian simulation experiment using a noisy mQSVT circuit, one is ought to benchmark this noisy implementation with respect to the ideal (noiseless) mQSVT circuit to determine the simulation performance.
Similar to XEB benchmarking that is used in RCS experiment~\cite{Arute_2019}, a benchmarking metric called the linear cross-entropy score (sXES) has been proposed in \cite{Dong_2022} to benchmark a noisy Hamiltonian simulation experiments.
The sXES metric can be used to compute the fidelity between the implemented noisy circuit and ideal circuit~\cite[eqn.(4)]{Dong_2022} by sampling the implemented device and performing some classical computations. 

\begin{definition}
    System linear cross-entropy score (sXES) of a noisy implementation of an mQSVT circuit given a unitary $U$ is defined as
    \begin{gather}\label{eqn:sxes}
        \mathrm{sXES}(U) = \sum_{x\neq 0^n} p(U,x) \, p_\mathrm{exp}(U,x)
    \end{gather}
    where $p_\mathrm{exp}(U,x)$ is the output probability from the noisy implementation of mQSVT circuit.
\end{definition}

How well the performance of an mQSVT circuit implementation as quantified by sXES is shown to determine whether it is able to solve a problem called the system linear cross-entropy heavy output generation (sXHOG).
Essentially, sXHOG problem requires the implemented circuit to output a fixed amount of distinct $n$-bit strings that have sufficiently large output probabilities from the ideal mQSVT circuit.

\begin{definition}
    System linear cross-entropy heavy output generation (sXHOG) problem:
    For a given mQSVT circuit with $n+1$ qubit unitary $U$ and some $b>1$, output $k$ many distinct non-zero $n$ bit strings $\{x_1,\dots,x_k\}\subseteq\{0,1\}^n\backslash\{0^n\}$ such that they satisfy
    \begin{gather}
        \frac{1}{k} \sum_{j=1}^k p(U,x_j) > \frac{b}{2^n} \;.
    \end{gather}
\end{definition}

sXHOG can be solved by sampling from an MQSVT implementation with sufficiently high sXES score (as it implies a sufficiently high circuitsfidelity).
However, the hardness of solving sXHOG with some high probability depending on a suitable choice of $k$ and $b$ using a classical algorithm has been asserted under an assumption that any polynomial-time classical algorithm cannot approximate the ideal mQSVT circuit output probability better (in mean squared error) than a trivial uniform approximation up to some negligible margin. 
This assumption is what is called the System linear cross-entropy quantum threshold assumption (sXQUATH).

\begin{definition}\label{def:sxquath}
    System linear cross-entropy quantum threshold assumption (sXQUATH) conjecture:
    There is no polynomial-time classical algorithm $\mathtt{C}$ such that the following holds for an mQSVT circuit for all $n\in\N$
    \begin{gather}
        \textup{sXQ} = 2^{2n} \bigg( \E_{U,X}\Big[\Big(p(U,X)-\frac{1}{2^n}\Big)^2\Big] - \E_{U,X}\Big[\Big(p(U,X)-q_{\mathtt{C}}(U,X)\Big)^2\Big] \bigg) = \Omega\Big(\frac{1}{2^n}\Big)
    \end{gather}
    where $n$ is the number of input/output qubits, and $p(U,X)$ is the probability of output $X$ on mQSVT circuit for a given unitary $U$, and $q_{\mathtt{C}}(U,X)$ is the probability of classical algorithm $\mathtt{C}$ simulating an mQSVT circuit outputting $X$ for a given unitary $U$.
    The expectations are over uniformly distributed random variable $X\in\{0,1\}^n\backslash\{0^n\}$ and over $n+1$ qubits Haar random unitaries $U$.
\end{definition}

Similar to XQUATH, sXQUATH conjecture states that there is no classical algorithm approximation of the output probability of an mQSVT circuit that has a mean squared error that is $2^{-n}$ less than the mean squared error of a trivial approximation.

\section{Reductions between threshold assumptions}\label{App:thershold_assumptions}

Here we discuss the complexity-theoretic relationship between XQUATH, sXQUATH, and its precursor, QUATH (Quantum threshold assumption).
The latter is related to another sampling-based quantum supremacy verification method called Heavy Output Generation (HOG)~\cite{aaronson2016complexitytheoretic}, which is a precursor to XHOG and sXHOG.
Here we again consider sampling from $n$ qubit unitary random circuit $U$ with all zero input and probability of outputting some string $x$ is $p(x,U) = |\<x|U|0^n\>|^2$.

\begin{definition}
    Heavy output generation (HOG) problem:
    For a random $n$ qubit unitary circuit $U$ drawn from some distribution $\mathcal{D}$, output $k$ many $n$ bit strings $\{x_1,\dots,x_k\}\subseteq\{0,1\}^n$ such that $\frac{2}{3}$ of their probabilities $p(x_1,U),\dots,p(x_k,U)$ are larger than the median of all the output probabilities $\textup{median}_{U\sim\mathcal{D}}\{p(x,U) : x\in\{0,1\}^n\}$.
\end{definition}

Compared to XHOG and sXHOG, HOG is clearly the most distinct as it requires that a certain proportion of the output strings to have probability larger than a value (the median) instead of requiring the sum of the probabilities to be larger than some value (as in XHOG and sXHOG).
However, similarly to XHOG and sXHOG, the hardness of classically solving HOG is also based on a conjecture called the Quantum threshold assumption (QUATH)~\cite{aaronson2016complexitytheoretic}, which as opposed to XQUATH and sXQUATH, is a conjecture for a decision problem.

\begin{definition}
    Quantum threshold assumption (QUATH):
    For some distribution $\mathcal{D}$ over $n$ qubit unitary circuit $U$, there is no polynomial-time classical algorithm taking a description of $U$ as input and outputs "Yes" if $p(0,U)$ is larger than $\textup{median}_{U\sim\mathcal{D}}\{p(x,U) : x\in\{0,1\}^n\}$ and "No" otherwise, with probability $\frac{1}{2} + \Omega(2^{-n})$ over $\mathcal{D}$.
\end{definition}

To obtain a better understanding on what features does a robust benchmarking scheme must have, one can gain insights by looking into the relationship between QUATH and XQUATH and sXQUATH.
Namely, whether one can be reduced from another is unclear.
Also, as XQUATH and sXQUATH for sublinear depth random unitaries has been refuted by using the Pauli path algorithms, it is interesting to further investigate whether the Pauli path algorithm can also disprove QUATH.

Even for XQUATH and sXQUATH, which bear most similarity among the three, their relationship does not seem to be straightforward.
As mentioned, the difficulty here is due to multiple copies of random unitary in an mQSVT circuit. 
Hence one would need to find some nontrivial correspondence between algorithms that approximates output probabilities a circuit containing multiple copies of random unitary $U$ and those with only a single $U$.
If we consider random unitaries made out of Haar-random two-qubit unitaries used in this work and in~\cite{aharonov2023polynomial}, we expect that this relationship between XQUATH and sXQUATH can be better understood by a deeper investigation into the relationship between two-qubit Haar-random unitary moment matrices of moment $t=2$ and $t>2$.

\section{Haar-random unitary moment matrix}\label{App:haar_unitary_moment_matrix}

As we will see later on how Pauli path algorithm is used to refute XQUATH and sXQUATH, an important ingredient is the expectation of $t$ Pauli matrices over product of $t$ Haar-random unitaries
\begin{gather}
    \E_V\Big[ V^{\otimes t} p_1 \otimes\dots\otimes p_t (V^\dag)^{\otimes t} \Big] \;,
\end{gather}
where $V$ is a $l$ qubit unitary matrix and $p_1,\dots,p_t \in \mathcal{P}_l$.
This quantity is extensively studied in~\cite{harrow2009random} and used in~\cite{Aharonov_Gao_Landau_Liu_Vazirani_2022} to show the properties of the transition amplitudes in the Fourier coefficients of Pauli paths.
For our purposes we focus on the case where $l=2$, i.e. $p_1,\dots,p_t \in \mathcal{P}_2$ are normalized two qubit Pauli matrices and the unitaries $V$ are two-qubit unitaries.

Now consider a matrix indexed by $t$ two-qubit Paulis
\begin{align}\label{eqn:haar_unitary_moment_matrix}
    G(\mathbf{p};\mathbf{q}) = \tr\bigg( \E_V\Big[ V^{\otimes t} p_1 \otimes\dots\otimes p_t (V^\dag)^{\otimes t} \Big] q_t \otimes\dots\otimes q_t \bigg) \;.
\end{align}
If one considers a random $N$-qubit circuit $U$ constructed from Haar-random two-qubit unitaries (as in~\cite{aharonov2023polynomial}), then the expectation of Pauli path transition amplitudes can be expressed as a product of expectation of transition amplitudes $\E_V[\llangle p |V|p' \rrangle] = \E_V[\tr(V p V^\dag p')] = G(p;p')$ for two-qubit Haar random unitary $V$ and normalized two-qubit paulis $p,p'$.
This is precisely the key analysis being done in analyzing the Pauli path algorithm for RCS~\cite{aharonov2023polynomial}

\section{Pauli path for mQSVT circuit}\label{App:mQSVT_pauli_path}

We may define a Pauli path $\mathbf{s}$ through an mQSVT circuit as a sequence of sub-paths $\mathbf{s}^{(k)} = s_1^{(k)},\dots,s_{d_U+1}^{(k)}$ and $\Tilde{\mathbf{s}}^{(k)} = \Tilde{s}_{d_U+1}^{(k)},\dots,\Tilde{s}_1^{(k)}$ that goes through unitary $U$ and $U^\dag$ in the $k$-th block, respectively.
Since phases $\varphi_1,\dots,\varphi_{2d+1}$ in an mQSVT circuit are fixed, we can instead simply absorb the rotation gates $R_z(\varphi_1)$ and $R_z(\varphi_{2d+1})$ at the beginning and end of the circuit into the preparation and measurement, respectively, in the mQSVT Pauli path simulation.
So for a single Pauli path $\mathbf{s}$ through a mQSVT circuit with unitary $U$ gives us a Fourier coefficient of the form
\begin{gather}
    F(U,\mathbf{s},x) = \llangle0x|\Tilde{s}_1^{(d)}\rrangle \bigg( \prod_{k=1}^d F_k(U,\mathbf{s}) \bigg) \llangle s_1^{(1)}|0^{n+1}\rrangle
\end{gather}
where we define
\begin{equation}
\begin{gathered}
    F_1(U,\mathbf{s}) = f_2(U^\dag,\mathbf{s})\, \llangle \Tilde{s}_{d_U+1}^{(k)}|R_z(\varphi_2)|s_{d_U+1}^{(k)}\rrangle\, f_1(U,\mathbf{s}) \;, \\
    F_k(U,\mathbf{s}) = f_{2k}(U^\dag,\mathbf{s})\, \llangle \Tilde{s}_{d_U+1}^{(k)}|R_z(\varphi_{2k})|s_{d_U+1}^{(k)}\rrangle\, f_{2k-1}(U,\mathbf{s})\, \llangle s_1^{(k)}|R_z(\varphi_{2k-1})|\Tilde{s}_1^{(k-1)}\rrangle \;, \\
    f_{2k-1}(U,\mathbf{s}) = \llangle s_{d_U+1}^{(k)}|U_{d_U}|s_{d_U}^{(k)} \rrangle \dots \llangle s_2^{(k)}|U_1| s_1^{(k)} \rrangle \;, \\
    f_{2k}(U^\dag,\mathbf{s}) = \llangle \Tilde{s}_1^{(k)}|U_1^\dag|\Tilde{s}_2^{(k)} \rrangle \dots \llangle \Tilde{s}_{d_U}^{(k)}|U_{d_U}^\dag| \Tilde{s}_{d_U+1}^{(k)} \rrangle
\end{gathered}
\end{equation}
where $U_j$ is the $j$-th layer of $U$ for $j\in\{1,\dots,d_U\}$.
Hence we can explicitly write an mQSVT Pauli path $\mathbf{s}$ in terms of its sub-paths as
\begin{gather}
    \mathbf{s} = \mathbf{s}^{(1)},\Tilde{\mathbf{s}}^{(1)},\mathbf{s}^{(2)},\Tilde{\mathbf{s}}^{(2)},\dots, \mathbf{s}^{(d)},\Tilde{\mathbf{s}}^{(d)} \;,
\end{gather}
which consists of $2d(d_U+1)$ many $n+1$-qubit Paulis.
Hence in terms of Pauli paths, the probability of output $x$ from an mQSVT circuit with unitary $U$ is
\begin{align}\label{eqn:outut_proba_pauli_path_expansion}
\begin{split}
    p(U,x) &= \sum_{\mathbf{s}} F(U,\mathbf{s},x) \;.
\end{split}
\end{align}

\section{Pauli path for single-block mQSVT circuit}\label{App:single_block_mQSVT_pauli_path}

For the special case of $d=1$, the probability of an mQSVT circuit with $n+1$ qubit unitary $U$ outputting string $x$ is
\begin{equation}
\begin{aligned}
    p(U,x) &= \big|\<0x| \big( R(\varphi_2)\otimes I^{\otimes n}\big) U^\dag \big( R(\varphi_1)\otimes I^{\otimes n}\big) U \big( R(\varphi_0)\otimes I^{\otimes n}\big) |0^{n+1}\>\big|^2 \\
    &= \big| \<0x| U^\dag \big( R(\varphi_1)\otimes I^{\otimes n}\big) U |0^{n+1}\> \big|^2 \\
    &= \sum_{\mathbf{s}} F(U,\mathbf{s},x)
\end{aligned}
\end{equation}
where we have Pauli path $\mathbf{s} = s_1,\dots,s_{d_U+1},\Tilde{s}_{d_U+1},\dots,\Tilde{s}_1$ and Fourier coefficient
\begin{equation}
\begin{aligned}
    F(U,\mathbf{s},x) &= \llangle 0x|\Tilde{s}_1 \rrangle \llangle\Tilde{s}_1|U_1^\dag|\Tilde{s}_2\rrangle \dots \llangle\Tilde{s}_{d_U}|U_1^\dag|\Tilde{s}_{d_U+1}\rrangle \llangle\Tilde{s}_{d_U+1}|R(\varphi_1)|s_{d_U+1}\rrangle \\
    &\quad \llangle s_{d_U+1}|U_{d_U}|s_{d_U}\rrangle \dots \llangle s_2|U_1|s_1\rrangle \llangle s_1|0^{n+1}\rrangle \;.
\end{aligned}
\end{equation}

We employ the circuit architecture of the random unitary circuit in~\cite{aharonov2023polynomial} to each copy of random unitary $U$ in the mQSVT circuit.
This unitary $U$ consists of two qubit gates, each drawn independently from the two-qubit unitary Haar measure.
Every qubit register goes through exactly one two-qubit unitary in each layer $U_1,\dots,U_{d+1}$.
No geometric locality on which two qubit registers any two-qubit gate operates on is assumed in this architecture. 
However since each qubit must go through a two-qubit gate in each layer we assume that the number of qubits $n+1$ is even so that in layer $U_j$ there are $(n+1)/2$ two-qubit gates $V_{j,1},\dots,V_{j,(n+1)/2}$.

First we look at the expectation of a Fourier coefficient for a single Pauli path for a single-block mQSVT circuit over Haar-random unitary $U$
\begin{equation}\label{eqn:expetation_fourier_mqsvt}
\begin{aligned}
    \E_U\Big[ F(U,\mathbf{s},x) \Big] &=  \llangle 0x|\Tilde{s}_1 \rrangle \llangle\Tilde{s}_{d_U+1}|R(\varphi_1)|s_{d_U+1}\rrangle \llangle s_1|0^{n+1}\rrangle \\
    &\quad \prod_{j=1}^{d_U} \prod_{i=1}^{(n+1)/2} \E_{V_{j,i}}\Big[ \tr\Big( V_{j,i}^{\otimes 2} (s_{j,i} \otimes \Tilde{s}_{j,i}) (V_{j,i}^\dag)^{\otimes 2} (s_{j+1,i} \otimes \Tilde{s}_{j+1,i}) \Big) \Big] \;.
\end{aligned}
\end{equation}
Note that the Haar-random expectation of a Fourier coefficient for the single-block mQSVT circuit involves a second-order Haar-random unitary moment, as there are two copies of every two-qubit gate $V_{j,i}$ in the circuit.
In contrast, the expectation of a random unitary circuit involves only the first-order moment
\begin{gather}
    \E_U\Big[ f(U,\mathbf{s},x) \Big] = \llangle x|s_{d_U+1} \rrangle \llangle s_1|0^{n+1}\rrangle 
    \prod_{j=1}^{d_U} \prod_{i=1}^{(n+1)/2} \E_{V_{j,i}}\Big[ \tr\Big( V_{j,i} s_{j,i} V_{j,i}^\dag s_{j+1,i} \Big) \Big] \;,
\end{gather}
where $f(U,\mathbf{s},x)$ is the Fourier coefficient of unitary circuit $U$ with output $x$ and Pauli path $\mathbf{s} = s_1,\dots,s_{d_U+1}$.
Although $\E_U[f(U,\mathbf{s},x)] = 0$ for any Pauli path $\mathbf{s}$ containing a non-identity Pauli, this is not the case for mQSVT circuits.
To identify the value of $\E_U[ F(U,\mathbf{s},x)]$ we need to look at the entries of the Haar unitary moment matrix $G$ in~\eqref{eqn:haar_unitary_moment_matrix} for $t=2$ which is shown in~\cite{harrow2009random} to be
\begin{equation}
\begin{aligned}\label{eqn:order_2_haar_moment}
        G( s,\Tilde{s} ; s',\Tilde{s}' ) &= \E_V\Big[ \tr\Big( V^{\otimes2} ( s \otimes \Tilde{s} ) (V^\dag)^{\otimes2} (s' \otimes \Tilde{s}') \Big) \Big] \\
        &= 
        \begin{cases}
            0 \quad&\textup{if $s\neq\Tilde{s}$ or $s'\neq\Tilde{s}'$}\\
            0 \quad&\textup{if exactly one of $s$ or $s'$ is $I\otimes I/2$}\\
            0 \quad&\textup{if exactly one of $\Tilde{s}$ or $\Tilde{s}'$ is $I\otimes I/2$}\\
            1 \quad&\textup{if $s=\Tilde{s}=s'=\Tilde{s}' = I\otimes I/2$}\\
            \frac{1}{15} \quad&\textup{if $s=\Tilde{s}\neq I\otimes I/2$ and $s'=\Tilde{s}'\neq I\otimes I/2$}
        \end{cases} \;.
\end{aligned}    
\end{equation}
By substituting these values into~\eqref{eqn:expetation_fourier_mqsvt} we can see that $\E_U[ F(U,\mathbf{s},x)] = 0$ if and only if (1) there is an input or output Pauli to some two-qubit gate $V_{j,i}$ that differs in the two unitaries in the mQSVT circuits, or (2) only either an input or output Pauli to some two-qubit gate $V_{j,i}$ is identity, i.e. there exists some $j\in[d_U]$ and $i\in[\frac{n+1}{2}]$ such that $s_{j,i}\neq \Tilde{s}_{j,i}$ or $s_{j+1,i}\neq \Tilde{s}_{j+1,i}$ or only one of $s_{j,i},s_{j+1,i}$ is identity or only one of $\Tilde{s}_{j,i},\Tilde{s}_{j+1,i}$ is identity.

Recall that in the Pauli path simulation for random unitary circuit in~\cite{aharonov2023polynomial}, Fourier coefficient of different Pauli path are orthogonal, i.e.
\begin{gather}
    \E_U\Big[ f(U,\mathbf{r},x) f(U,\mathbf{s},x) \Big] = 0
\end{gather}
for $\mathbf{r}\neq\mathbf{s}$.
This Fourier coefficient orthogonality condition does not hold for mQSVT circuit in general.
For a single-block mQSVT circuit the Fourier coefficient expectation above can be expanded as
\begin{equation}\label{eqn:expetation_two paths_fourier_mqsvt}
\begin{aligned}
    & \E_U\Big[ F(U,\mathbf{r},x) F(U,\mathbf{s},x) \Big] \\
    &= \llangle 0x|\Tilde{r}_1 \rrangle \llangle 0x|\Tilde{s}_1 \rrangle 
    \llangle\Tilde{r}_{d_U+1}|R(\varphi_1)|r_{d_U+1}\rrangle \llangle\Tilde{s}_{d_U+1}|R(\varphi_1)|s_{d_U+1}\rrangle 
    \llangle r_1|0^{n+1}\rrangle \llangle s_1|0^{n+1}\rrangle \\
    &\quad \prod_{j=1}^{d_U} \prod_{i=1}^{(n+1)/2} \E_{V_{j,i}}\Big[ \tr\Big( V_{j,i}^{\otimes 4} (r_{j,i} \otimes \Tilde{r}_{j,i} \otimes s_{j,i} \otimes \Tilde{s}_{j,i}) (V_{j,i}^\dag)^{\otimes 4} (r_{j+1,i} \otimes \Tilde{r}_{j+1,i} \otimes s_{j+1,i} \otimes \Tilde{s}_{j+1,i}) \Big) \Big]
\end{aligned}
\end{equation}
where the expectations are over two qubit unitaries $V_{j,i}$ and $r_{j,i} , \Tilde{r}_{j,i} , s_{j,i} , \Tilde{s}_{j,i} , r_{j+1,i} , \Tilde{r}_{j+1,i} , s_{j+1,i} , \Tilde{s}_{j+1,i}$ are normalized two qubit Paulis which at the inputs and outputs of $V_{j,i}$.
The two qubit unitary expectation terms above are the entries of Haar unitary moment matrix $G$ in~\eqref{eqn:haar_unitary_moment_matrix} for $t=4$ which can be expanded in terms of the Weingarten functions $\mathrm{Wg}$ as
\begin{equation}
\begin{aligned}\label{eqn:order_4_haar_moment_weingarten}
        G( r,\Tilde{r},s,\Tilde{s} ; r',\Tilde{r}',s',\Tilde{s}' ) &= \E_V\Big[ \tr\Big( V^{\otimes4} ( r \otimes \Tilde{r} \otimes s \otimes \Tilde{s} ) (V^\dag)^{\otimes4} (r' \otimes \Tilde{r}' \otimes s' \otimes \Tilde{s}') \Big) \Big] \\
        &= \sum_{\tau,\pi} \mathrm{Wg}(\tau\pi) \tr(W_{\tau^{-1}} r \otimes \Tilde{r} \otimes s \otimes \Tilde{s}) \, \tr(W_\pi r' \otimes \Tilde{r}' \otimes s' \otimes \Tilde{s}') \;,
\end{aligned}    
\end{equation}
where $\tau,\pi$ are permutations in $S_4$ and $\tau\pi$ is the composition of permutations $\pi$ and $\tau$.

To determine the value of expectation over the Fourier coefficients of mQSVT circuit in~\eqref{eqn:expetation_two paths_fourier_mqsvt}, we need to analyze the terms in the sum in~\eqref{eqn:order_4_haar_moment_weingarten}.
First, we need to determine for which Paulis the value of~\eqref{eqn:order_4_haar_moment_weingarten} is zero.
Here we define the \textit{cyclic partition} of a permutation $\pi$ on ordered set $B$ as the partition of $B$ where each partition is a subset of $B$ that is permuted by a cycle in the cyclic decomposition of $\pi$.
For the term $\tr(W_\pi r \otimes \Tilde{r} \otimes s \otimes \Tilde{s})$, the set being permuted is $\{r , \Tilde{r} , s , \Tilde{s}\}$.
So, if the cyclic decomposition is $\pi=(134)(2)$ then its cyclic partition is $\{ \{r,s,\Tilde{s}\}, \{\Tilde{r}\} \}$.
For simplicity in what follows, let us denote the unnormalized two-qubit Paulis by the capital letters of their corresponding normalized two-qubit Paulis, e.g. $R$ for $r$ and $\Tilde{S}$ for $s$, so we can write eqn~\eqref{eqn:order_4_haar_moment_weingarten} as
\begin{equation}
\begin{aligned}
        G( r,\Tilde{r},s,\Tilde{s} ; r',\Tilde{r}',s',\Tilde{s}' ) &= \frac{1}{2^8} \sum_{\tau,\pi} \mathrm{Wg}(\tau\pi) \tr(W_{\tau^{-1}} R \otimes \Tilde{R} \otimes S \otimes \Tilde{S}) \, \tr(W_\pi R' \otimes \Tilde{R}' \otimes S' \otimes \Tilde{S}') \;.
\end{aligned}    
\end{equation}
Now as a first step to evaluate some of the relevant values of~\eqref{eqn:order_4_haar_moment_weingarten} we show the following lemma.

\begin{lemma}\label{lem:nonzero_pauli_permutation}
    For $k$ two-qubit Paulis $Q_1,\dots,Q_k$ and permutation operator $W_\pi$ for permutation $\pi\in S_k$, we have $\tr(W_\pi Q_1 \otimes \dots \otimes Q_k) \neq 0$ if and only if the product of Paulis in each set in the cyclic partition of $\pi$ on $\{Q_1,\dots,Q_k\}$ is equal to $I\otimes I$.
    Moreover, let the number of cyclic partition of $\pi$ be $c_\pi$.
    Then, if the product of Paulis in each partition of $\pi$ is identity then $\tr(W_\pi Q_1 \otimes \dots \otimes Q_k) =  \tr(I\otimes I)^{c_\pi} = 4^{c_\pi}$.
\end{lemma}
\begin{proof}
    Suppose that permutation $\pi\in S_k$ has $c_\pi$ cyclic partitions so we can express it as $\pi = \pi_1\dots\pi_{c_\pi}$ where $\pi_1,\dots,\pi_{c_\pi}$ are disjoint cycles.
    So we can rewrite $\tr(W_\pi Q_1 \otimes \dots \otimes Q_k)$ as
    \begin{equation}
    \begin{aligned}
        \tr(W_\pi Q_1 \otimes \dots \otimes Q_k) &= \prod_{j=1}^{c_\pi} \tr(W_{\pi_j} \Bar{Q}_j)
    \end{aligned}
    \end{equation}
    where $\Bar{Q}_j$ is the product of two-qubit Paulis in $\{Q_1,\dots,Q_k\}$ in the $j$-th partition of $\pi$.
    For example, if $\pi_j=(1436)$ then $\Bar{Q}_j = Q_1\otimes Q_3 \otimes Q_4 \otimes Q_6$.
    Suppose that the $j$-th cycle is $\pi_j=(a_1,\dots,a_l)$ where $a_1,\dots,a_l$ are distinct numbers in $1,\dots,k$ indicating the indices of the two-qubit Paulis in the $j$-th partition.
    Then we have
    \begin{equation}
    \begin{aligned}
        \tr(W_{\pi_j} \Bar{Q}_j) = \tr\Big( Q_{a_1} Q_{\pi_j(a_1)} Q_{\pi_j^2(a_1)} \dots Q_{\pi_j^{l-1}(a_1)} \Big)
    \end{aligned}
    \end{equation}
    which is equal to 0 if $Q_{a_1} Q_{\pi_j(a_1)} Q_{\pi_j^2(a_1)} \dots Q_{\pi_j^{l-1}(a_1)}$ is not equal to identity $I\otimes I$ and otherwise $\tr(W_{\pi_j} \Bar{Q}_j) = \tr(I\otimes I) = 4$.
    The claim in the lemma follows by following the same argument for each of the $c_\pi$ partitions.
\end{proof}

Values of the Weingarten function only depends on the cyclic partition of the permutation $\tau\pi$, particularly thesize of each partition.
As we are interested only in permutations over four elements, there are only five possible cyclic partitions we denote as $[4]$, $[3,1]$, $[2,2]$, $[2,1,1]$, and $[1,1,1,1]$, where the numbers inside the square brackets are the size of each partition.
Now for simplicity, we do a slight abuse of notation in writing the Weingarten function to only denote the cyclic partition of the permutation in its argument and list its values
\begin{equation}
\begin{aligned}
    \mathrm{Wg}([4]) &= \frac{-20}{\Delta} \\
    \mathrm{Wg}([3,1]) &= \frac{29}{\Delta} \\
    \mathrm{Wg}([2,2]) &= \frac{32}{\Delta} \\
    \mathrm{Wg}([2,1,1]) &= \frac{-48}{\Delta} \\
    \mathrm{Wg}([1,1,1,1]) &= \frac{134}{\Delta}
\end{aligned}
\end{equation}
where $\Delta = 4(4^2-1)(4^2-4)(4^2-9) = 20160$.
For a more general formula of the Weingarten function on $\mathcal{S}_4$ see the appendix of~\cite{roberts2017chaos}.

For $\mathcal{S},\mathcal{S}'\subseteq \mathcal{S}_4$, let
\begin{equation}
\begin{aligned}
    &\Bar{G}(r,\Tilde{r},s,\Tilde{s} ; r',\Tilde{r}',s',\Tilde{s}') \\
    &= \sum_{\tau\in \mathcal{S}, \pi\in \mathcal{S}'} \mathrm{Wg}(\tau\pi) \tr(W_{\tau^{-1}} r \otimes \Tilde{r} \otimes s \otimes \Tilde{s}) \, \tr(W_\pi r' \otimes \Tilde{r}' \otimes s' \otimes \Tilde{s}') \\
    &= \frac{1}{2^8} \sum_{\tau\in \mathcal{S}, \pi\in \mathcal{S}'} \mathrm{Wg}(\tau\pi) \tr(W_{\tau^{-1}} R \otimes \Tilde{R} \otimes S \otimes \Tilde{S}) \, \tr(W_\pi R' \otimes \Tilde{R}' \otimes S' \otimes \Tilde{S}') \;,
\end{aligned}
\end{equation}
where $\mathcal{S}$ and $\mathcal{S}'$ are subsets of $\mathcal{S}_4$ defined as
\begin{equation}
\begin{aligned}
    \mathcal{S} &= \Big\{ \pi\in \mathcal{S}_4 : \tr(W_\pi R \otimes \Tilde{R} \otimes S \otimes \Tilde{S}) \neq 0 \Big\} \\
    \mathcal{S}' &= \Big\{ \pi\in \mathcal{S}_4 : \tr(W_\pi R' \otimes \Tilde{R}' \otimes S' \otimes \Tilde{S}') \neq 0 \Big\} \;.
\end{aligned}
\end{equation}
Therefore we get
\begin{equation}
    G(r,\Tilde{r},s,\Tilde{s} ; r',\Tilde{r}',s',\Tilde{s}') = \Bar{G}(r,\Tilde{r},s,\Tilde{s} ; r',\Tilde{r}',s',\Tilde{s}')
\end{equation}
hence $\Bar{G}$ consist of the non-zero terms in $G$ which we can get by using Lemma~\ref{lem:nonzero_pauli_permutation}, namely the only terms that we need to compute to get the value of $G$ for two qubit Paulis in its argument that we are interested in.

Now we consider the values of $G$ for specific Paulis $r,\Tilde{r},r',\Tilde{r}'$:
\begin{enumerate}
    \item $r=\Tilde{r}=r'=\Tilde{r}'=ZI/2$,
    \item $r=\Tilde{r}=r'=\Tilde{r}'=II/2$, and
    \item $r=ZZ/2$ and $\Tilde{r}=r'=\Tilde{r}'=ZI/2$,
\end{enumerate}
which relevance will be clear in our proof that sXQUATH does not hold for single-block mQSVT.
The different values that $G$ can take for the aforementioned two qubit Paulis are
\begin{align}
    \begin{split}
        &G(ZI/2,ZI/2,s,\Tilde{s} ; ZI/2,ZI/2,s',\Tilde{s}') \\
        &= 
        \begin{cases}
            \frac{90048}{2^8\Delta} \quad&\textup{if $s=\Tilde{s}=s'=\Tilde{s}'=ZI/2$} \\
            \frac{31680}{2^8\Delta} \quad&\textup{if $s=\Tilde{s}\notin\{ZI/2,II/2\}$ and $s'=\Tilde{s}'\notin\{ZI/2,II/2\}$} \\
            \frac{18368}{2^8\Delta} \quad&\textup{if $s=\Tilde{s}=ZI/2$ and $ s'=\Tilde{s}'\notin\{ZI/2,II/2\}$} \\
            \frac{18368}{2^8\Delta} \quad&\textup{if $s=\Tilde{s}\notin\{ZI/2,II/2\}$ and $ s'=\Tilde{s}'=ZI/2$} \\
            \frac{602048}{2^8\Delta} \quad&\textup{if $s=\Tilde{s}=s'=\Tilde{s}'=II/2$} \\
            0 \quad&\textup{otherwise}
        \end{cases} \label{eqn:haar_4_moment_ZI}
    \end{split}\\
    &G(II/2,II/2,s,\Tilde{s} ; II/2,II/2,s',\Tilde{s}') = 
    \begin{cases}
        1 \quad&\textup{if $s=\Tilde{s}=s'=\Tilde{s}'=II/2$} \\
        \frac{31680}{2^8\Delta} \quad&\textup{if $s=\Tilde{s}\neq II/2$ and $s'=\Tilde{s}'\neq II/2$} \\
        0 \quad&\textup{otherwise}
    \end{cases} \label{eqn:haar_4_moment_II} \\
    \begin{split}
        &G(ZZ/2,ZI/2,s,\Tilde{s} ; ZI/2,ZI/2,s',\Tilde{s}') \\
        &= 
        \begin{cases}
            \frac{95232}{2^8\Delta} \quad&\textup{if $\{s,\Tilde{s}\}=\{ZZ/2,ZI/2\}$ and $s'=\Tilde{s}'=ZI/2$} \\
            \frac{-121920}{2^8\Delta} \quad&\textup{if $\{s,\Tilde{s}\}=\{PX/2,PY/2\}$ and $s'=\Tilde{s}'=ZI/2$} \\
            \frac{73664}{2^8\Delta} \quad&\textup{if $\{s,\Tilde{s}\}=\{IZ/2,II/2\}$ and $s'=\Tilde{s}'=ZI/2$} \\
            \frac{-6656}{2^8\Delta} \quad&\textup{if $\{s,\Tilde{s}\}=\{ZZ/2,ZI/2\}$ and $s'=\Tilde{s}'\notin\{ZI/2,II/2\}$} \\
            \frac{12224}{2^8\Delta} \quad&\textup{if $\{s,\Tilde{s}\}=\{PX/2,PY/2\}$ and $s'=\Tilde{s}'\notin\{ZI/2,II/2\}$} \\
            \frac{960}{2^8\Delta} \quad&\textup{if $\{s,\Tilde{s}\}=\{IZ/2,II/2\}$ and $s'=\Tilde{s}'\notin\{ZI/2,II/2\}$} \\
            0 \quad&\textup{otherwise}
        \end{cases} \;, \label{eqn:haar_4_moment_ZZ_ZI}
    \end{split}
\end{align}
where $P\in\{I,X,Y,Z\}$ in eqn.~\eqref{eqn:haar_4_moment_ZZ_ZI}.
Note that each of the values above lies in the interval $[-1,1]$ since $2^8\Delta = 5160960$.

\section{sXQUATH does not hold for single-block mQSVT circuit}\label{App:sxquath_false}

In this section we give a detailed proof of Theorem 2 stated in the main manuscript.
First, consider Pauli path algorithm that gives a classical approximation $q(U,x)$ of the probability $p(U,x)$ for single-block mQSVT circuit with $n+1$ qubit unitary $U$ with depth $d_U$ outputting an $n$ bit string $x\neq 0^n$, where
\begin{gather}\label{eqn:pauli_path_classical_estimation}
    q(U,x) = \frac{1}{2^n} + F(U,\mathbf{r},x) \;.
\end{gather}
Here we have Pauli path $\mathbf{r} = (Z\otimes I^l \otimes Z \otimes I^{\otimes n-l-1} , Z\otimes I^{\otimes n} , \dots , Z\otimes I^{\otimes n})$ such that the input Paulis into the two qubit gate in the first layer of unitary $U$ that operates on the first qubit is $Z\otimes Z$.
For example, if in the first layer the two qubit gate that operates on the first qubit couples it with the 5th qubit then $Z\otimes I^l \otimes Z \otimes I^{\otimes n-l-1} = Z\otimes I^3 \otimes Z \otimes I^{\otimes n-4}$.

As the architecture that we consider puts every qubit register through a two qubit gate in each layer $j\in[d_U]$, we assume that $n+1$ is even. 
This is the same architecture used in~\cite{aharonov2023polynomial}.
For notational simplicity in the decomposition of $U$ into two qubit gates $\{V_{j,i} : j\in[d_U], i\in[(n+1)/2]\}$ as in the previous section, we denote the two qubit gate operating on the first qubit register in each layer $j$ as $V_{j,1}$.
So for Pauli path $\mathbf{r}$, these are the only two qubit gates in $U$ (therefore in the entire mQSVT circuit) that has non-identity input and output Paulis in the Pauli path approximation.

Now if by using $q(U,x)$ in sXQUATH (Definition~\ref{def:sxquath}) as an approximation to $p(U,x)$ and by expanding the squares and simplifying the terms we get
\begin{align}
    \textup{sXQ} &= 2^{2n} \bigg( \E_{U,X}\Big[\Big(p(U,X)-\frac{1}{2^n}\Big)^2\Big] - \E_{U,X}\Big[\Big(p(U,X)-q(U,X)\Big)^2\Big] \bigg) \\
    &= 2^{2n} \E_{U,X}\bigg[ \frac{1}{2^{2n}} - \frac{2}{2^n}p(U,X) - q(U,X)^2 + 2p(U,X)q(U,X) \bigg] \\
    &= \frac{2^{2n}}{2^n-1} \sum_{x\neq 0^n} \E_U\bigg[ \frac{1}{2^{2n}} - \frac{2}{2^n}p(U,x) - q(U,x)^2 + 2p(U,x)q(U,x) \bigg] \\
    &= \frac{2^{2n}}{2^n-1} \sum_{x\neq 0^n} \E_U\bigg[ \frac{1}{2^{2n}} - \frac{2}{2^n}p(U,x) - q(U,x)^2 + \frac{2}{2^n}p(U,x) + 2p(U,x)F(U,\mathbf{r},x) \bigg] \\
    &= \frac{2^{2n}}{2^n-1} \sum_{x\neq 0^n} \E_U\bigg[ \frac{1}{2^{2n}} - \Big(\frac{1}{2^{2n}}+\frac{2}{2^n}F(U,\mathbf{r},x)+F(U,\mathbf{r},x)^2\Big) + 2p(U,x)F(U,\mathbf{r},x) \bigg] \\
    &= \frac{2^{2n}}{2^n-1} \sum_{x\neq 0^n} \bigg( -\frac{2}{2^n}\E_U\Big[F(U,\mathbf{r},x)\Big] - \E_U\Big[F(U,\mathbf{r},x)^2\Big] + 2\E_U\Big[p(U,x)F(U,\mathbf{r},x) \Big] \bigg) \;.
\end{align}
Furthermore if we express the probability $p(U,x)$ in terms of Pauli paths, namely $p(U,x) = \sum_{\mathbf{s}} F(U,\mathbf{s},x)$ we get
\begin{align}
    \E_U\Big[p(U,x)F(U,\mathbf{r},x) \Big] &= \sum_{\mathbf{s}} \E_U\Big[ F(U,\mathbf{s},x)\, F(U,\mathbf{r},x) \Big]
\end{align}
where the first equality is obtained by expanding $p(U,x)$ in terms of Pauli paths (eqn.\eqref{eqn:outut_proba_pauli_path_expansion}).
Hence the sXQUATH expression above becomes
\begin{equation}\label{eqn:sxquath_simplified}
\begin{aligned}
    \textup{sXQ} &= \frac{2^{2n}}{2^n-1} \sum_{x\neq 0^n} \bigg( -\frac{2}{2^n}\E_U\Big[F(U,\mathbf{r},x)\Big] - \E_U\Big[F(U,\mathbf{r},x)^2\Big] + 2\sum_{\mathbf{s}\neq\mathbf{r}} \E_U\Big[ F(U,\mathbf{s},x)\, F(U,\mathbf{r},x) \Big] \bigg) \;.
\end{aligned}
\end{equation}
Now we treat each terms in the sum over $x$ above separately.

For the first term in the outer sum of eqn.~\eqref{eqn:sxquath_simplified}, we have $\E_U[F(U,\mathbf{r},x)]=0$ because by expanding it in terms of the Pauli paths as in eqn.~\eqref{eqn:expetation_fourier_mqsvt} we get the term
\begin{equation}
\begin{aligned}
    \frac{1}{2^4} \E_V\Big[\tr\Big(V^{\otimes 2} ZI \otimes ZZ (V^\dag)^{\otimes 2} ZI\otimes ZI \Big)\Big] = \frac{1}{2^4} G(ZI , ZZ ; ZI , ZI ) = 0
\end{aligned}
\end{equation}
by eqn.~\eqref{eqn:order_2_haar_moment}.
For the second term $\E_U[F(U,\mathbf{r},x)^2]$, we use eqn.~\eqref{eqn:expetation_two paths_fourier_mqsvt} with $\mathbf{r}=\mathbf{s}$ and evaluate the product of Haar-random expectation values using \eqref{eqn:haar_4_moment_ZZ_ZI}, \eqref{eqn:haar_4_moment_ZI}, and \eqref{eqn:haar_4_moment_II} to get
\begin{equation}\label{eqn:Haar_moment_product_fourier_squared}
\begin{aligned}
    & \prod_{j=1}^{d_U} \prod_{i=1}^{(n+1)/2} \E_{V_{j,i}}\Big[ \tr\Big( V_{j,i}^{\otimes 4} (r_{j,i} \otimes \Tilde{r}_{j,i} \otimes s_{j,i} \otimes \Tilde{s}_{j,i}) (V_{j,i}^\dag)^{\otimes 4} (r_{j+1,i} \otimes \Tilde{r}_{j+1,i} \otimes s_{j+1,i} \otimes \Tilde{s}_{j+1,i}) \Big) \Big] \\
    &= G( ZZ/2 , ZI/2 , ZZ/2 , ZI/2 ; ZI/2 , ZI/2 , ZI/2 , ZI/2 ) \\ &\quad \Big( G( ZI/2 , ZI/2 , ZI/2 , ZI/2 ; ZI/2 , ZI/2 , ZI/2 , ZI/2 ) \Big)^{d_U-1} \\ &\quad \Big( G( II/2 , II/2 , II/2 , II/2 ; II/2 , II/2 , II/2 , II/2 ) \Big)^{(d_U-1) + \frac{n-1}{2}} \\
    &= \frac{95232}{2^8\Delta} \bigg(\frac{90048}{2^8\Delta}\bigg)^{d_U-1} \\
    &= (2^8\Delta)^{-d_U} \times 95232 \times 90048^{d_U-1} \;.
\end{aligned}
\end{equation}
For shorthand let $\gamma = (2^8\Delta)^{-d_U} \times 95232 \times 90048^{d_U-1}$, so by substituting this into eqn.~\eqref{eqn:expetation_two paths_fourier_mqsvt} we get
\begin{equation}\label{eqn:fourier_coefficient_squared_expectation_pauli_sim}
\begin{aligned}
    &\E_U[F(U,\mathbf{r},x)^2] = \gamma \llangle 0x|\Tilde{r}_1 \rrangle^2 
    \llangle\Tilde{r}_{d_U+1}|R(\varphi_1)|r_{d_U+1}\rrangle^2
    \llangle r_1|0^{n+1}\rrangle^2 \\
    &\quad= \frac{\gamma}{2^{4(n+1)}} 
    \underbrace{\<0x|Z\otimes I^{\otimes n}|0x\>^2}_{=1} \,
    \underbrace{\tr\big( R(\varphi_1) Z R(\varphi_1)^\dag Z \big)^2}_{=2^2} \underbrace{\tr(I^{\otimes n})^2}_{2^{2n}} \,
    \underbrace{\<0^{n+1}|Z\otimes Z\otimes I^{\otimes n-1}|0^{n+1}\>^2}_{=1} \\
    &\quad= \frac{\gamma}{2^{2(n+1)}} 
\end{aligned}
\end{equation}
because $\<y|I|y\>=1$ for $y\in\{0,1\}$ and $\<0|Z|0\>=1$ and because $R(\varphi_1)$ and $Z$ commute.
Since this is true for all $x\neq0^n$ then by pulling in the sum over $x\neq0^n$ in eqn.~\eqref{eqn:sxquath_simplified} inside the brackets we get the term
\begin{equation}\label{eqn:sum_expectation_fourier_final_value}
    \sum_{x\neq0^n} \E_U[F(U,\mathbf{r},x)^2] =  \frac{2^n-1}{2^{2(n+1)}} \gamma \;.
\end{equation}

Now we are left with the $\sum_{\mathbf{s}\neq\mathbf{r}} \E_U[ F(U,\mathbf{s},x)\, F(U,\mathbf{r},x) ]$ term in eqn.~\eqref{eqn:sxquath_simplified}, where we need to evaluate this for every Pauli path $\mathbf{s}$ such that $\mathbf{s}\neq\mathbf{r}$.
Again we expand $\E_U[ F(U,\mathbf{s},x)\, F(U,\mathbf{r},x) ]$ using eqn.~\eqref{eqn:expetation_two paths_fourier_mqsvt} which contains product of Haar-random expectations
\begin{equation}\label{eqn:product_haar_moments_nonequal_s}
\begin{aligned}
    \gamma_\mathbf{s} &= \prod_{j=1}^{d_U} \prod_{i=1}^{(n+1)/2} \E_{V_{j,i}}\Big[ \tr\Big( V_{j,i}^{\otimes 4} (r_{j,i} \otimes \Tilde{r}_{j,i} \otimes s_{j,i} \otimes \Tilde{s}_{j,i}) (V_{j,i}^\dag)^{\otimes 4} (r_{j+1,i} \otimes \Tilde{r}_{j+1,i} \otimes s_{j+1,i} \otimes \Tilde{s}_{j+1,i}) \Big) \Big] \\
    &= G( ZZ/2 , ZI/2 , s_{1,1} , \Tilde{s}_{1,1} ; ZI/2 , ZI/2 , s_{2,1} , \Tilde{s}_{2,1} ) \\ &\quad \bigg( \prod_{j=2}^{d_U} G( ZI/2 , ZI/2 , s_{j,1} , \Tilde{s}_{j,1} ; ZI/2 , ZI/2 , s_{j+1,1} , \Tilde{s}_{j+1,1} ) \bigg) \\ &\quad \bigg( \prod_{j=1}^{d_U} \prod_{i=2}^{(n+1)/2} G( II/2 , II/2 , s_{j,i} , \Tilde{s}_{j,i} ; II/2 , II/2 , s_{j+1,i} , \Tilde{s}_{j+1,i} ) \bigg) \;.
\end{aligned}
\end{equation}
Note that $\gamma_\mathbf{s}$ is non-zero if and only if Pauli path $\mathbf{s}$ is such that all instance of $G$ in it is non-zero, which can be identified using \eqref{eqn:haar_4_moment_ZZ_ZI}, \eqref{eqn:haar_4_moment_ZI}, and \eqref{eqn:haar_4_moment_II}.
Moreover, the values of $\llangle s_{1,1} | 00 \rrangle$ and $\llangle 0x_1 | \Tilde{s}_{1,1} \rrangle$ also determines which $\mathbf{s}$ gives a non-zero $\E_U[ F(U,\mathbf{s},x)\, F(U,\mathbf{r},x) ]$.
Their values for all $\mathbf{s}$ such that $\gamma_\mathbf{s}\neq0$ (see eqn.\eqref{eqn:haar_4_moment_ZZ_ZI}) are
\begin{equation}
\begin{aligned}
    \llangle s_{1,1} | 00 \rrangle =
    \begin{cases}
        0 \quad&\textup{if $s_{1,1}\in\{PX/2,PY/2\}$}\\
        \frac{1}{2} \quad&\textup{if $s_{1,1}\in\{ZI/2,ZZ/2,IZ/2,II/2\}$}
    \end{cases} \\
    \llangle 0x_1 | \Tilde{s}_{1,1} \rrangle =
    \begin{cases}
        0 \quad&\textup{if $\Tilde{s}_{1,1}\in\{PX/2,PY/2\}$}\\
        \frac{1}{2} \quad&\textup{if $\Tilde{s}_{1,1}\in\{ZI/2,II/2\}$}\\
        \frac{1}{2} \quad&\textup{if $\Tilde{s}_{1,1}\in\{ZZ/2,IZ/2\}$ and $x_1=0$}\\
        -\frac{1}{2} \quad&\textup{if $\Tilde{s}_{1,1}\in\{ZZ/2,IZ/2\}$ and $x_1=1$}
    \end{cases} \;,
\end{aligned}
\end{equation}
which implies that we only need to consider $\mathbf{s}$ (i.e. those that gives nonzero $\gamma_\mathbf{s}$ and $\llangle s_{1,1} | 00 \rrangle$ and $\llangle 0x_1 | \Tilde{s}_{1,1} \rrangle$) such that
\begin{equation}\label{eqn:pauli_paths_nonzero_expectation}
    (s_{1,1},\Tilde{s}_{1,1})\in\{(ZI/2,ZZ/2),(ZZ/2,ZI/2),(IZ/2,II/2),(II/2,IZ/2)\} \;.
\end{equation}

Now by doing similar calculation as in~\eqref{eqn:fourier_coefficient_squared_expectation_pauli_sim},
\begin{equation}
\begin{aligned}
    &\E_U[ F(U,\mathbf{s},x)\, F(U,\mathbf{r},x) ] \\
    &= \gamma_\mathbf{s} \llangle 0x|\Tilde{r}_1 \rrangle \llangle 0x|\Tilde{s}_1 \rrangle 
    \llangle\Tilde{r}_{d_U+1}|R(\varphi_1)|r_{d_U+1}\rrangle \llangle\Tilde{s}_{d_U+1}|R(\varphi_1)|s_{d_U+1}\rrangle 
    \llangle r_1|0^{n+1}\rrangle \llangle s_1|0^{n+1}\rrangle \\
    &= \frac{\gamma_\mathbf{s}}{2^{2n}}\llangle 0x_1 | \Tilde{s}_{1,1} \rrangle \, \llangle s_{1,1} | 00 \rrangle \;.
\end{aligned}
\end{equation}
Then by pulling the sum over $x\neq0^n$ inside the brackets in eqn.~\eqref{eqn:sxquath_simplified}, we get
\begin{equation}\label{eqn:nonequal_fourier_coefficient_overlap_expectation_pauli_sim}
\begin{aligned}
    \sum_{x\neq0^n} \sum_{\mathbf{s}\neq\mathbf{r}} \E_U[ F(U,\mathbf{s},x)\, F(U,\mathbf{r},x) ] &= \frac{1}{2^{2n}} \sum_{x\neq0^n} \sum_{\mathbf{s}\neq\mathbf{r}} \gamma_\mathbf{s} \llangle 0x_1 | \Tilde{s}_{1,1} \rrangle \, \llangle s_{1,1} | 00 \rrangle \\
    &= \frac{1}{2^{2n}} \sum_{\mathbf{s}\in\mathcal{S}} \gamma_\mathbf{s} \llangle s_{1,1} | 00 \rrangle \, \Big( (2^{n-1}-1) \llangle 00 | \Tilde{s}_{1,1} \rrangle + 2^{n-1} \llangle 01 | \Tilde{s}_{1,1} \rrangle \Big)
\end{aligned}
\end{equation}
where $\mathcal{S}$ is the set of Pauli paths such that $\mathbf{s}\neq\mathbf{r}$ and $\gamma_\mathbf{s}\neq0$ and $\llangle s_{1,1} | 00 \rrangle\neq0$ and $\llangle 0x_1 | \Tilde{s}_{1,1} \rrangle\neq0$.
Now consider a fixed $s_2,\dots,s_{d_U+1},\Tilde{s}_{d_U+1},\dots,\Tilde{s}_2 \neq r_2,\dots,r_{d_U+1},\Tilde{r}_{d_U+1},\dots,\Tilde{r}_2$.
By using eqn.~\eqref{eqn:haar_4_moment_ZZ_ZI}, let $\gamma_\mathbf{s}(Z^3)$ be the product of Haar moments for a Pauli path $\mathbf{s}$ such that $(s_{1,1},\Tilde{s}_{1,1})\in\{(ZI/2,ZZ/2),(ZZ/2,ZI/2)\}$ and $\gamma_\mathbf{s}(Z^1)$ the product of Haar moments for a Pauli path $\mathbf{s}$ such that $(s_{1,1},\Tilde{s}_{1,1})\in\{(IZ/2,II/2),(II/2,IZ/2)\}$.
For such Pauli paths $s_2,\dots,s_{d_U+1},\Tilde{s}_{d_U+1},\dots,\Tilde{s}_2$ we have
\begin{equation}\label{eqn:exp_fourier_sum_nonequal_s}
\begin{aligned}
    &\sum_{s_{1,1},\Tilde{s}_{1,1}} \gamma_\mathbf{s} \llangle s_{1,1} | 00 \rrangle \, \Big( (2^{n-1}-1) \llangle 00 | \Tilde{s}_{1,1} \rrangle + 2^{n-1} \llangle 01 | \Tilde{s}_{1,1} \rrangle \Big) \\
    &= \bigg(\sum_{s_{1,1},\Tilde{s}_{1,1}} \frac{\gamma_\mathbf{s}(2^{n-1}-1)}{4}\bigg) + \bigg( \frac{\gamma_{\mathbf{s}(Z^3)}}{2} \Big(\frac{2^{n-1}}{2}-\frac{2^{n-1}}{2}\Big) + \frac{\gamma_{\mathbf{s}(Z^1)}}{2} \Big( \frac{2^{n-1}}{2}-\frac{2^{n-1}}{2} \Big) \bigg) \\
    &= \sum_{s_{1,1},\Tilde{s}_{1,1}} \frac{\gamma_\mathbf{s}(2^{n-1}-1)}{4} \;,
\end{aligned}
\end{equation}
where the sum is over the pair $s_{1,1},\Tilde{s}_{1,1}$ in~\eqref{eqn:pauli_paths_nonzero_expectation}.
Note that the first equality is due to $\llangle s_{1,1} | 00 \rrangle = \llangle 00 | \Tilde{s}_{1,1} \rrangle = \frac{1}{2}$.
For $s_2,\dots,s_{d_U+1},\Tilde{s}_{d_U+1},\dots,\Tilde{s}_2 = r_2,\dots,r_{d_U+1},\Tilde{r}_{d_U+1},\dots,\Tilde{r}_2$, similarly we get
\begin{equation}\label{eqn:exp_fourier_sum_equal_s}
\begin{aligned}
    &\sum_{s_{1,1},\Tilde{s}_{1,1}} \gamma_\mathbf{s} \llangle s_{1,1} | 00 \rrangle \, \Big( (2^{n-1}-1) \llangle 00 | \Tilde{s}_{1,1} \rrangle + 2^{n-1} \llangle 01 | \Tilde{s}_{1,1} \rrangle \Big) \\
    &= \bigg(\sum_{s_{1,1},\Tilde{s}_{1,1}} \frac{\gamma_\mathbf{s}(2^{n-1}-1)}{4}\bigg) + \bigg( \frac{\gamma_{\mathbf{s}(Z^3)}}{2} \Big(-\frac{2^{n-1}}{2}\Big) + \frac{\gamma_{\mathbf{s}(Z^1)}}{2} \Big( \frac{2^{n-1}}{2}-\frac{2^{n-1}}{2} \Big) \bigg) \\
    &= \frac{\gamma_{\mathbf{s}(Z^3)} (2^{n-1}-1)}{2} + \frac{\gamma_{\mathbf{s}(Z^1)} (2^{n-1}-1)}{2} - \frac{\gamma_{\mathbf{s}(Z^3)} 2^{n-1}}{4} \\
    &= \frac{\gamma_{\mathbf{s}(Z^3)} (2^{n-1}-2)}{4} + \frac{\gamma_{\mathbf{s}(Z^1)} (2^{n-1}-1)}{2} \;,
\end{aligned}
\end{equation}
where the sum in the first and second line are over the pair $s_{1,1},\Tilde{s}_{1,1}$ in~\eqref{eqn:pauli_paths_nonzero_expectation} except for $(s_{1,1},\Tilde{s}_{1,1}) = (ZZ/2,ZI/2)$ because $(r_{1,1},\Tilde{r}_{1,1}) = (ZZ/2,ZI/2)$ and we require that $\mathbf{s}\neq\mathbf{r}$.
Hence by substituting eqn.\eqref{eqn:exp_fourier_sum_nonequal_s} and eqn.~\eqref{eqn:exp_fourier_sum_equal_s} back into eqn.~\eqref{eqn:nonequal_fourier_coefficient_overlap_expectation_pauli_sim} we obtain
\begin{equation}\label{eqn:simplified_nonequal_fourier_coefficient_overlap_expectation_pauli_sim}
\begin{aligned}
    &\sum_{x\neq0^n}\sum_{\mathbf{s}\neq\mathbf{r}} \E_U[ F(U,\mathbf{s},x)\, F(U,\mathbf{r},x) ] \\
    &= \frac{1}{2^{2n}} \bigg( \sum_{s_{1,1},\Tilde{s}_{1,1}} \sum_{\mathbf{s}_{2:}\in\mathcal{S}_{2:}\backslash\{\mathbf{r}_{2:}\}} \frac{\gamma_\mathbf{s}(2^{n-1}-1)}{4} \bigg) + \frac{1}{2^{2n}} \bigg( \frac{\gamma_{\mathbf{s}(Z^3)} (2^{n-1}-2)}{4} + \frac{\gamma_{\mathbf{s}(Z^1)} (2^{n-1}-1)}{2} \bigg)
\end{aligned}
\end{equation}
where we denote $\mathbf{s}_{2:} = s_2,\dots,s_{d_U+1},\Tilde{s}_{d_U+1},\dots,\Tilde{s}_2$ and $\mathcal{S}_{2:}\backslash\{\mathbf{r}_{2:}\}$ is the set of Pauli paths $\mathbf{s}$ such that $\mathbf{s}_{2:} \neq r_2,\dots,r_{d_U+1},\Tilde{r}_{d_U+1},\dots,\Tilde{r}_2$.
Also, here the product of Haar moments $\gamma_\mathbf{s}(Z^3)$ and $\gamma_\mathbf{s}(Z^1)$ are such that $\mathbf{s}_{2:} = r_2,\dots,r_{d_U+1},\Tilde{r}_{d_U+1},\dots,\Tilde{r}_2$.

We now analyze the value of the $\gamma_\mathbf{s}$'s.
For Pauli path $\mathbf{s}$ satisfying $\mathbf{s}_{2:} = r_2,\dots,r_{d_U+1},\Tilde{r}_{d_U+1},\dots,\Tilde{r}_2$ a calculation similar to eqn.~\eqref{eqn:Haar_moment_product_fourier_squared} by using~\eqref{eqn:haar_4_moment_ZZ_ZI} gives us
\begin{equation}
\begin{aligned}
    \gamma_\mathbf{s}(Z^3) &= \bigg(\frac{90048}{2^8\Delta}\bigg)^{d_U-1} G(ZZ/2,ZI/2,IZ/2,ZZ/2 ; ZI/2,ZI/2,ZI/2,ZI/2) \\
    &= (2^8\Delta)^{-d_U} \times 95232 \times 90048^{d_U-1} \\
    &\textup{and} \\
    \gamma_\mathbf{s}(Z^1) &= \bigg(\frac{90048}{2^8\Delta}\bigg)^{d_U-1} G(ZZ/2,ZI/2,IZ/2,II/2 ; ZI/2,ZI/2,ZI/2,ZI/2) \\
    &= \bigg(\frac{90048}{2^8\Delta}\bigg)^{d_U-1} G(ZZ/2,ZI/2,II/2,IZ/2 ; ZI/2,ZI/2,ZI/2,ZI/2) \\
    &= (2^8\Delta)^{-d_U} \times 73664 \times 90048^{d_U-1} \;.
\end{aligned}
\end{equation}
Hence, the second term in~\eqref{eqn:simplified_nonequal_fourier_coefficient_overlap_expectation_pauli_sim} is
\begin{equation}\label{eqn:simplified_nonequal_fourier_coefficient_overlap_expectation_pauli_sim_second}
\begin{aligned}
    \frac{1}{2^{2n}} \bigg( \frac{\gamma_{\mathbf{s}(Z^3)} (2^{n-1}-2) + \gamma_{\mathbf{s}(Z^1)} (2^n-2)}{4} \bigg) &= \frac{1}{2^{2n}}\bigg( \frac{(2^{n-1}\times242560 - 2\times168896) 90048^{d_U-1}}{4(2^8\Delta)^{d_U}} \bigg) \;,
\end{aligned}
\end{equation}
which is positive for all $n\geq2$.

For Pauli paths $\mathbf{s}$ such that $\mathbf{s}_{2:} \neq r_2,\dots,r_{d_U+1},\Tilde{r}_{d_U+1},\dots,\Tilde{r}_2$, let us denote
\begin{equation}
\begin{gathered}
    g_1(s,\Tilde{s},s',\Tilde{s}') = G(ZZ/2,ZI/2,s,\Tilde{s} ; ZI/2,ZI/2,s',\Tilde{s}') \\
    g_2(s,\Tilde{s},s',\Tilde{s}') = G(ZI/2,ZI/2,s,\Tilde{s} ; ZI/2,ZI/2,s',\Tilde{s}') \\
    g_3(s,\Tilde{s},s',\Tilde{s}') = G(II/2,II/2,s,\Tilde{s} ; II/2,II/2,s',\Tilde{s}') \;,
\end{gathered}
\end{equation}
so that we can write $\gamma_\mathbf{s}$ as
\begin{equation}\label{eqn:shorthand_haar_moment_product_nonequal_s}
\begin{aligned}
    \gamma_\mathbf{s} = g_1(s_{1,1},\Tilde{s}_{1,1},s_{2,1},\Tilde{s}_{2,1}) 
    \bigg( \prod_{j=2}^{d_U} g_2(s_{j,1},\Tilde{s}_{j,1},s_{j+1,1},\Tilde{s}_{j+1,1}) \bigg)
    \bigg( \prod_{j=1}^{d_U} \prod_{i=2}^{(n+1)/2} g_3(s_{j,i},\Tilde{s}_{j,i},s_{j+1,i},\Tilde{s}_{j+1,i}) \bigg) \;.
\end{aligned}
\end{equation}
For such Pauli path, in the product of Haar moments $\gamma_\mathbf{s}$ (see eqn.~\eqref{eqn:product_haar_moments_nonequal_s}) there are:
\begin{enumerate}
    \item 2 possible $s_{1,1},\Tilde{s}_{1,1},s_{2,1},\Tilde{s}_{2,1}$ such that $g_1(s_{1,1},\Tilde{s}_{1,1},s_{2,1},\Tilde{s}_{2,1}) = 95232 (2^8\Delta)^{-1}$

    \item 2 possible $s_{1,1},\Tilde{s}_{1,1},s_{2,1},\Tilde{s}_{2,1}$ such that $g_1(s_{1,1},\Tilde{s}_{1,1},s_{2,1},\Tilde{s}_{2,1}) = 73664 (2^8\Delta)^{-1}$

    \item 28 possible $s_{1,1},\Tilde{s}_{1,1},s_{2,1},\Tilde{s}_{2,1}$ such that $g_1(s_{1,1},\Tilde{s}_{1,1},s_{2,1},\Tilde{s}_{2,1}) = -6656 (2^8\Delta)^{-1}$

    \item 28 possible $s_{1,1},\Tilde{s}_{1,1},s_{2,1},\Tilde{s}_{2,1}$ such that $g_1(s_{1,1},\Tilde{s}_{1,1},s_{2,1},\Tilde{s}_{2,1}) = 960 (2^8\Delta)^{-1}$.
\end{enumerate}
For each $j\in\{2,\dots,d_U\}$, there are:
\begin{enumerate}
    \item 1 possible $s_{j,1},\Tilde{s}_{j,1},s_{j+1,1},\Tilde{s}_{j+1,1}$ such that $g_2(s_{j,1},\Tilde{s}_{j,1},s_{j+1,1},\Tilde{s}_{j+1,1}) = 90048 (2^8\Delta)^{-1}$

    \item $14\times14=196$ possible $s_{j,1},\Tilde{s}_{j,1},s_{j+1,1},\Tilde{s}_{j+1,1}$ such that $g_2(s_{j,1},\Tilde{s}_{j,1},s_{j+1,1},\Tilde{s}_{j+1,1}) = 31680 (2^8\Delta)^{-1}$

    \item 28 possible $s_{j,1},\Tilde{s}_{j,1},s_{j+1,1},\Tilde{s}_{j+1,1}$ such that $g_2(s_{j,1},\Tilde{s}_{j,1},s_{j+1,1},\Tilde{s}_{j+1,1}) = 18368 (2^8\Delta)^{-1}$

    \item 1 possible $s_{j,1},\Tilde{s}_{j,1},s_{j+1,1},\Tilde{s}_{j+1,1}$ such that $g_2(s_{j,1},\Tilde{s}_{j,1},s_{j+1,1},\Tilde{s}_{j+1,1}) = 602048 (2^8\Delta)^{-1}$.
\end{enumerate}
For each $j\in\{1,\dots,d_U\}$ and $i\in\{2,\dots,(n+1)/2\}$ there are:
\begin{enumerate}
    \item 1 possible $s_{j,i},\Tilde{s}_{j,i},s_{j+1,i},\Tilde{s}_{j+1,i}$ such that $g_3(s_{j,i},\Tilde{s}_{j,i},s_{j+1,i},\Tilde{s}_{j+1,i}) = 1$

    \item $15\times15=225$ possible $s_{j,i},\Tilde{s}_{j,i},s_{j+1,i},\Tilde{s}_{j+1,i}$ such that $g_3(s_{j,i},\Tilde{s}_{j,i},s_{j+1,i},\Tilde{s}_{j+1,i}) = 31680 (2^8\Delta)^{-1}$.
\end{enumerate}

Note that among all possible values of $g_1,g_2,g_3$ there is only a single negative value, namely when $g_1$ takes the argument of $\{s,\Tilde{s}\}=\{ZZ/2,ZI/2\}$ and $s'=\Tilde{s}'\notin\{ZI/2,II/2\}$ in which it gives $-6656(2^8\Delta)^{-1}$.
Hence two products of $g_2$'s and $g_3$'s in~\eqref{eqn:shorthand_haar_moment_product_nonequal_s} are positive for any Pauli path $\mathbf{s}$.
Now if we fix all Paulis in $\mathbf{s}$ except for $s_{1,1},\Tilde{s}_{1,1},s_{2,1},\Tilde{s}_{2,1}$ then let
\begin{equation}
\begin{aligned}
    \xi_\mathbf{s} &= \sum_{s_{1,1},\Tilde{s}_{1,1}} \sum_{s_{2,1},\Tilde{s}_{2,1}} \frac{\gamma_\mathbf{s}(2^{n-1}-1)}{4} \\
    &= \frac{2^{n-1}-1}{4} \bigg( \prod_{j=2}^{d_U} g_2(s_{j,1},\Tilde{s}_{j,1},s_{j+1,1},\Tilde{s}_{j+1,1}) \bigg) \bigg( \prod_{j=1}^{d_U} \prod_{i=2}^{(n+1)/2} g_3(s_{j,i},\Tilde{s}_{j,i},s_{j+1,i},\Tilde{s}_{j+1,i}) \bigg) \\
    &\quad \times \Big( 2\times95232 + 2\times73664 + 28\times(-6656) + 28\times960 \Big) (2^8\Delta)^{-1} \\
    &= \frac{2^{n-1}-1}{4} \frac{178304}{2^8\Delta} \bigg( \prod_{j=2}^{d_U} g_2(s_{j,1},\Tilde{s}_{j,1},s_{j+1,1},\Tilde{s}_{j+1,1}) \bigg) \bigg( \prod_{j=1}^{d_U} \prod_{i=2}^{(n+1)/2} g_3(s_{j,i},\Tilde{s}_{j,i},s_{j+1,i},\Tilde{s}_{j+1,i}) \bigg) \;,
\end{aligned}
\end{equation}
which is positive.
Therefore for any choice of Paulis that determines the value of the $g_2$'s and $g_3$'s, namely $\{(s_{j,1},\Tilde{s}_{j,1},s_{j+1,1},\Tilde{s}_{j+1,1})\}_{j=2}^{d_U}$ and $\{(s_{j,i},\Tilde{s}_{j,i},s_{j+1,i},\Tilde{s}_{j+1,i})\}_{j=1,i=2}^{d_U,(n+1)/2}$, the first term of eqn.~\eqref{eqn:simplified_nonequal_fourier_coefficient_overlap_expectation_pauli_sim} are positive.
Hence we can obtain a lower bound
\begin{equation}\label{eqn:simplified_nonequal_fourier_coefficient_overlap_expectation_pauli_sim_first}
\begin{aligned}
    &\frac{1}{2^{2n}} \bigg( \sum_{\mathbf{s}\backslash\{s_{1,1},\Tilde{s}_{1,1},s_{2,1},\Tilde{s}_{2,1}\}} \sum_{s_{1,1},\Tilde{s}_{1,1}} \sum_{s_{2,1},\Tilde{s}_{2,1}} \frac{\gamma_\mathbf{s}(2^{n-1}-1)}{4} \bigg) \\
    &= \frac{1}{2^{2n}} \bigg( \sum_{\mathbf{s}\backslash\{s_{1,1},\Tilde{s}_{1,1},s_{2,1},\Tilde{s}_{2,1}\}} \xi_\mathbf{s} \bigg) \\
    &\geq \frac{1}{2^{2n}} \frac{2^{n-1}-1}{4} \frac{178304}{2^8\Delta} \bigg(\frac{602048}{2^8\Delta}\bigg)^{d_U-1}
\end{aligned}
\end{equation}
by simply picking a single term in the sum on the second line where Pauli path $\mathbf{s}$ is such that $\{(s_{j,1},\Tilde{s}_{j,1},s_{j+1,1},\Tilde{s}_{j+1,1})\}_{j=2}^{d_U}$ and $\{(s_{j,i},\Tilde{s}_{j,i},s_{j+1,i},\Tilde{s}_{j+1,i})\}_{j=1,i=2}^{d_U,(n+1)/2}$ are all identities $II/2$.
Note that this choice of $\mathbf{s}$ indeed satisfies $\mathbf{s}\neq\mathbf{r}$.

Finally by substituting eqn.~\eqref{eqn:simplified_nonequal_fourier_coefficient_overlap_expectation_pauli_sim_second} and eqn.~\eqref{eqn:simplified_nonequal_fourier_coefficient_overlap_expectation_pauli_sim_first} back into~\eqref{eqn:simplified_nonequal_fourier_coefficient_overlap_expectation_pauli_sim} we get a lower bound
\begin{equation}\label{eqn:sum_fourier_coefficient_pauli_sim_lower_bound}
\begin{aligned}
    &\sum_{x\neq0^n}\sum_{\mathbf{s}\neq\mathbf{r}} \E_U[ F(U,\mathbf{s},x)\, F(U,\mathbf{r},x) ] \\
    &\geq \frac{(2^{n-1}-1)\times178304\times602048^{d_U-1} + (2^{n-1}\times242560 - 2\times168896) 90048^{d_U-1}}{2^{2n+2}(2^8\Delta)^{d_U}}
\end{aligned}
\end{equation}
Putting eqn.~\eqref{eqn:sum_expectation_fourier_final_value} and eqn.~\eqref{eqn:sum_fourier_coefficient_pauli_sim_lower_bound} into sXQUATH value~\eqref{eqn:sxquath_simplified} gives us the final lower bound
\begin{equation}
\begin{aligned}
    \textup{sXQ} &= \frac{2^{2n}}{2^n-1} \sum_{x\neq 0^n} \bigg( -\frac{2}{2^n}\E_U\Big[F(U,\mathbf{r},x)\Big] - \E_U\Big[F(U,\mathbf{r},x)^2\Big] + 2\sum_{\mathbf{s}\neq\mathbf{r}} \E_U\Big[ F(U,\mathbf{s},x)\, F(U,\mathbf{r},x) \Big] \bigg) \\
    &\geq \frac{1}{2^n-1} \bigg( \frac{2^n-1}{2^2} \frac{95232\times90048^{d_U-1}}{(2^8\Delta)^{d_U}} \\
    &\quad + \frac{(2^{n-1}-1)\times178304\times602048^{d_U-1} + (2^{n-1}\times242560 - 2\times168896) 90048^{d_U-1}}{2^2(2^8\Delta)^{d_U}} \bigg) \\
    &\approx \frac{95232\times90048^{d_U-1}}{2^2(2^8\Delta)^{d_U}} + \frac{178304\times602048^{d_U-1} + 242560\times90048^{d_U-1}}{2^3(2^8\Delta)^{d_U}} \\
    &= c^{d_U}
\end{aligned}
\end{equation}
for some constant $c\in(0,1)$ independent of number of qubits $n$ and sub-circuit $U$ depth $d_U$.
By setting $d_U=o(n)$, this refutes sXQUATH for single-block mQSVT circuit.

\section{Spoofing cross-entropy benchmarking using Pauli path simulation}\label{App:spoofing_sXES}

Here we give a proof of Theorem 2 stated in the main manuscript by using Pauli path algorithm to spoof sXES benchmarking defined in eqn.~\eqref{eqn:sxes}, on average.
This is done by using the Pauli path algorithm used in the proof of Theorem 1 (explained in detail in Supplementary Material~\ref{App:sxquath_false}).
The average sXES score over random unitary $U$ is given by
\begin{gather}
    \mathrm{sXES} = \E_U[\mathrm{sXES}(U)] = \E_U\Big[\sum_{x\neq0^n} p(U,x)\, q(U,x)\Big] \;.
\end{gather}
For convenience we restate the approximation $q(U,x)$ in Supplementary Material~\ref{App:sxquath_false}
\begin{gather}
    q(U,x) = \frac{1}{2^n} + F(U,\mathbf{r},x) \;,
\end{gather}
for Pauli path $\mathbf{r} = (Z\otimes I^l \otimes Z \otimes I^{\otimes n-l-1} , Z\otimes I^{\otimes n} , \dots , Z\otimes I^{\otimes n})$.
Details of this Pauli path algorithm, i.e. which two-qubit unitary the input and output Paulis correspond to, are in the paragraphs surrounding eqn.\eqref{eqn:pauli_path_classical_estimation}.

By using this Pauli path simulation, we can write the sXES score as
\begin{equation}
\begin{aligned}
    \mathrm{sXES} &= \sum_{x\neq0^n} \E_U\Big[p(U,x)\, \Big(\frac{1}{2^n} + F(U,\mathbf{r},x)\Big) \Big] \\
    &= \sum_{x\neq0^n} \E_U\Big[ \frac{p(U,x)}{2^n}\Big] + \E_U\Big[p(U,x) \, F(U,\mathbf{r},x) \Big] \\
    &= \sum_{x\neq0^n} \E_U\Big[ \frac{p(U,x)}{2^n}\Big] + \sum_{\mathbf{s}\neq\mathbf{r}} \E_U\Big[F(U,\mathbf{s},x) \, F(U,\mathbf{r},x) \Big] \;.
\end{aligned}
\end{equation}
Note that we can lower bound the second term using eqn.~\eqref{eqn:sum_fourier_coefficient_pauli_sim_lower_bound}, restated below for convenience
\begin{equation}
\begin{aligned}
    &\sum_{x\neq0^n}\sum_{\mathbf{s}\neq\mathbf{r}} \E_U[ F(U,\mathbf{s},x)\, F(U,\mathbf{r},x) ] \\
    &\geq \frac{(2^{n-1}-1)\times178304\times602048^{d_U-1} + (2^{n-1}\times242560 - 2\times168896) 90048^{d_U-1}}{2^{2n+2}(2^8\Delta)^{d_U}} \;.
\end{aligned}
\end{equation}
Hence we can lower bound sXES score of this Pauli path simulation as
\begin{equation}
\begin{aligned}
    \mathrm{sXES} &\geq \frac{1}{2^n}\E_U\Big[\sum_{x\neq0^n}p(U,x)\Big] \\
        &\quad+\frac{(2^{n-1}-1)\times178304\times602048^{d_U-1} + (2^{n-1}\times242560 - 2\times168896) 90048^{d_U-1}}{2^{2n+2}(2^8\Delta)^{d_U}} \\
    &\approx \frac{1}{2^n}\Big(\E_U\Big[\sum_{x\neq0^n}p(U,x)\Big] + c^{d_U}\Big)
\end{aligned}
\end{equation}
for some constant $c\in(0,1)$ independent of $n$ and $d_U$.

\end{document}